\documentclass[pdflatex,sn-mathphys-num]{sn-jnl}

\usepackage{graphicx}%
\usepackage{multirow}%
\usepackage{amsmath,amssymb,amsfonts}%
\usepackage{amsthm}%
\usepackage{mathrsfs}%
\usepackage[title]{appendix}%
\usepackage{xcolor}%
\usepackage{textcomp}%
\usepackage{manyfoot}%
\usepackage{booktabs}%
\usepackage{placeins}
\usepackage{algorithm}%
\usepackage{algorithmicx}%
\usepackage{algpseudocode}%
\usepackage{listings}%
\usepackage{subcaption}

\theoremstyle{thmstyleone}%
\newtheorem{theorem}{Theorem}
\newtheorem{proposition}[theorem]{Proposition}%

\theoremstyle{thmstyletwo}%
\newtheorem{remark}{Remark}%
\newtheorem{lemma}{Lemma}
\theoremstyle{thmstylethree}%

\begin{document}

\title[Traveling Waves and Bumps in $p$-adic NNs]{Pseudo-Traveling Waves and Bumps in  Quantum and Classical Hierarchical Cellular Neural Networks}


\author*[1]{\fnm{W. A.} \sur{Z\'{u}\~{n}iga-Galindo}}\email{wilson.zunigagalindo@utrgv.edu}

\author[2]{\fnm{B. A.} \sur{Zambrano-Luna}}\email{bzambran@ualberta.ca}
\equalcont{These authors contributed equally to this work.}

\author[1]{\fnm{Chayapuntika} \sur{Indoung}}\email{chayapuntika.indoung01@utrgv.edu}
\equalcont{These authors contributed equally to this work.}

\affil*[1]{\orgdiv{School of Mathematical \& Statistical Sciences}, \orgname{University of Texas Rio Grande Valley}, \orgaddress{\street{One West University Blvd}, \city{Brownsville}, \postcode{78520}, \state{Texas}, \country{United States}}}

\affil[2]{\orgdiv{Department Of Mathematical and Statistical Sciences}, \orgname{University Of Alberta}, \orgaddress{\street{CAB 632}, \city{Edmonton},  \state{Alberta}, \country{Canada T6G 2G1}}}



\abstract{We study the existence of pseudo-traveling waves and bump solutions for two
classes of hierarchical cellular neural networks (CNNs) defined over the
ring of $p$-adic integers $\mathbb{Z}_{p}$. The first type is a $p$-adic CNN
described by a reaction-diffusion equation, while the second type is its
quantum analog obtained via Wick rotation. The $p$-adic CNNs are
hierarchical versions of the classical Chua-Yang CNNs; these networks have a
tree-like hierarchical architecture with infinitely many cells and hidden
layers. The states are governed by integro-differential equations on $%
\mathbb{Z}_{p}$. The $p$-adic traveling waves behave fundamentally
differently from their Archimedean counterparts. A traveling wave restricted
to a $p$-adic sphere yields a countably infinite collection of independent
patterns. We introduce the notion of pseudo-traveling waves as finite
truncations of this structure and prove their existence for both the
classical and quantum networks. We further establish the existence of
time-independent solutions (bumps) for both models. Our theoretical results
are complemented by numerical simulations that approximate
pseudo-traveling-wave solutions for quantum CNNs.
}




\maketitle

\section{Introduction}

In \cite{Zambrano-Zuniga-1}-\cite{Zambrano-Zuniga-2}, see also \cite{Zuniga et
al}, the first two authors introduced the $p$-adic cellular neural networks
(CNNs), which are hierarchical generalizations of the CNNs introduced \ by
Chua and Yang in 1980s, see e.g., \cite{Chua-Tamas}-\cite{Slavova}. The
$p$-adic numbers are naturally organized in a tree-like structure, which can
be used to construct models of hierarchical neural networks (NNs). The
$p$-adic CNNs have an infinite number of cells organized hierarchically in a
tree-like structure, with an infinite number of hidden layers. Intuitively,
these networks occur as limits of large hierarchical discrete CNNs.

Here, we \ consider hierarchical NNs whose topology can be codified using the
ring of $p$-adic integers $\mathbb{Z}_{p}$. Geometrically $\mathbb{Z}_{p}$ is
an infinite rooted tree, with infinite layers; the vertices at the $\left(
l+1\right)  $-th layer corresponds of $p$-adic numbers of the form
$a_{0}+a_{1}p+\ldots+a_{l}p^{l}$, where the $a_{i}$ belong to $\left\{
0,1,\ldots,p-1\right\}  $. In this paper, we consider $p$-adic CNNs whose
states are described by
\begin{equation}
\frac{\partial}{\partial t}u\left(  x,t\right)  =J\left(  \left\vert
x\right\vert _{p}\right)  \ast u\left(  x,t\right)  -u\left(  x,t\right)  +%
{\displaystyle\int\limits_{\mathbb{Z}_{p}}}
W\left(  \left\vert x-y\right\vert _{p}\right)  \phi\left(  u\left(
y,t\right)  \right)  dy+Z(\left\vert x\right\vert _{p}), \label{CNN-1}%
\end{equation}
where $x\in\mathbb{Z}_{p}$, $t\geq0$, $J:\left[  0,1\right]  \rightarrow
\mathbb{R}_{\geq0}$, and $\int_{\mathbb{Z}_{p}}J\left(  \left\vert
x\right\vert _{p}\right)  dx=1$, and $W$, $Z:\mathbb{R}\rightarrow\mathbb{R}$,
$\left\vert \cdot\right\vert _{p}$ denotes the $p$-adic norm, $dx$ is the Haar
measure of $\mathbb{Z}_{p}$, and and$\ u\left(  x,t\right)  $ is a real-valued
function. This function is interpreted as a measure of the neural activity (or state) at position $x$ (or neuron $x$) at time $t$. The above integro-differential equation is a reaction-diffusion equation. The diffusion\ part is
controlled by%
\begin{equation}
\frac{\partial}{\partial t}u\left(  x,t\right)  =J\left(  \left\vert
x\right\vert _{p}\right)  \ast u\left(  x,t\right)  -u\left(  x,t\right)  .
\label{Heat-1}%
\end{equation}
This is a heat equation on $\mathbb{Z}_{p}$, which means \ that the
fundamental solution is the transition probability of a Markov process on
$\mathbb{Z}_{p}$; see, e.g., \cite{Zuniga-PHYA}-\cite{Kochubei}.

By performing a Wick rotation in (\ref{Heat-1}), and taking $\Psi\left(
x,t\right)  =u\left(  x,it\right)  $, we obtain a free Schr\"{o}dinger
equation%
\[
i\frac{\partial}{\partial t}\Psi\left(  x,t\right)  =\Psi\left(  x,t\right)
-J\left(  \left\vert x\right\vert _{p}\right)  \ast\Psi\left(  x,t\right)  .
\]
Now, the quantum analog of (\ref{CNN-1}) is
\begin{equation}
i\frac{\partial}{\partial t}\Psi\left(  x,t\right)  =\Psi\left(  x,t\right)
-J\left(  \left\vert x\right\vert _{p}\right)  \ast\Psi\left(  x,t\right)  +%
{\displaystyle\int\limits_{\mathbb{Z}_{p}}}
W\left(  \left\vert x-y\right\vert _{p}\right)  \phi\left(  \Psi\left(
y,t\right)  \right)  dy+Z(x). \label{QNN-1}%
\end{equation}

In \cite{Zuniga-QNNs}, the authors introduced the quantum neural networks
(QNNs) of type (\ref{QNN-1}); they gave a detailed study of the discretization
of the new $p$-adic Schr\"{o}dinger equations, which allows the construction
of new QNNs on simple graphs. They also provided detailed numerical
simulations, offering a clear insight into the functioning of the new QNNs. At
a mathematical level, the authors showed the existence of global solutions for
the new $p$-adic Schr\"{o}dinger equations. This paper continues this work, by
studying the existence of special solutions (pseudo-traveling waves and bumps)
for equations (\ref{CNN-1}) and (\ref{QNN-1}).

This paper aims to initiate the study of traveling waves in NNs of types
(\ref{CNN-1}) and (\ref{QNN-1}). By traveling waves, we mean solutions of the
form $\psi\left(  \left\vert x\right\vert _{p}+vt\right)  $, $\psi
:\mathbb{R}\rightarrow\mathbb{C}$, where $x\in\mathbb{Z}_{p}$, $t\geq0$. If
$v=0$, we say that $\psi\left(  \left\vert x\right\vert _{p}\right)  $ is a
bump. In the standard case, the existence of traveling waves is reduced to the
solution of certain ODEs. This fact is not true in the $p$-adic case, which
makes it very challenging. We argue that the dificulty comes from the fact
that\ the topology of $\mathbb{Z}_{p}$ is radically different from the one of
$\mathbb{R}$; $\mathbb{Z}_{p}$ is homeomorphic to Cantor-like subset of
$\mathbb{R}$; see, e.g., \cite{V-V-Z}, \cite{Alberio et al}. Consider an
Archimedean (standard) traveling wave $\chi\left(  x+vt\right)  $, where
$x,t,v\in\mathbb{R}$, where $v$ is fixed. For $t=t_{0}$ fixed, the restriction
of the traveling wave\ to the set $\left\vert x\right\vert =a$ gives two
values $\chi\left(  \pm a+vt_{0}\right)  $, and consequently, \ the mentioned
restriction does not provide any relevant information about the traveling
wave. The $p$-adic case is completely different. The set $S_{-j}=\left\{
x\in\mathbb{Z}_{p};\left\vert x\right\vert _{p}=p^{-j}\right\}  $ is the
sphere centered at the origin with radius $p^{-j}$, with $j$ a non-negative
integer. This sphere is an infinite set with a positive Haar measure. Denote
by $1_{S_{-j}}\left(  x\right)  $ the characteristic function of $S_{-j}$.
With this notation,%
\begin{equation}
\psi\left(  \left\vert x\right\vert _{p}+vt\right)  =%
{\displaystyle\sum\limits_{j=0}^{\infty}}
\psi\left(  p^{-j}+vt\right)  1_{S_{-j}}\left(  x\right)  ,
\label{Traveling_Wave}%
\end{equation}
which means that the pattern associated with a traveling wave consists of a
countable collection of independent patterns controlled by the same
parameters. We warn the reader that (\ref{Traveling_Wave}) \ is not a series,
because the supports of the functions $1_{S_{-j}}\left(  x\right)  $ are
disjoint, for this reason, $\psi\left(  \left\vert x\right\vert _{p}%
+vt\right)  $ is an independent collection of patterns.

Motivated by the above discussion, we study the existence\ of pseudo-traveling
waves, which are functions of the form%
\begin{equation}%
{\displaystyle\sum\limits_{j=0}^{M}}
\psi_{j}\left(  p^{-j}+vt\right)  1_{S_{-j}}\left(  x\right)  +\Omega
(p^{M+1}\left\vert x\right\vert _{p})\psi_{M+1}\left(  vt\right)  ,
\label{Traveling_Wave_2}%
\end{equation}
where $\Omega(p^{M+1}\left\vert x\right\vert _{p})$ is the characteristic
function of the ball
\[
B_{-M-1}=\left\{  x\in\mathbb{Z}_{p};\left\vert x\right\vert _{p}\leq
p^{-M-1}\right\}  ,
\]
where the functions $\psi_{j}:\mathbb{R}\rightarrow\mathbb{C}$, \ $j=0,1,\ldots,M,M+1$ are differentiable, and $v$ is a positive real number. We show that
(\ref{CNN-1}) and (\ref{QNN-1}) admit pseudo-traveling waves; see Theorem
\ref{Theorem-3A}, and the remarks that follow it. On the other hand, our
numerical simulations are based on pseudo-traveling waves.

Notice that $\lim_{M\rightarrow\infty}\Omega(p^{M+1}\left\vert x\right\vert
_{p})=0$, so taking the limit $M\rightarrow\infty$ in (\ref{Traveling_Wave_2}%
), we obtain the traveling wave $\sum_{j=0}^{\infty}\psi_{j}\left(
p^{-j}+vt\right)  1_{S_{-j}}\left(  x\right)  $, however we do \ not know if
this function (or distribution) is a solution of (\ref{CNN-1}) or
(\ref{QNN-1}). So the existence of traveling waves for NNs of types
(\ref{CNN-1}) and (\ref{QNN-1}) is an open problem. We developed a numerical technique to approximate solutions of the form \ref{Traveling_Wave_2}. This technique allows us to provide numerical approximations for pseudo-traveling waves of QCNNs of type \ref{QNN-1}.

Finally, we show the existence of time-independent \ solutions (bumps) for NNs of types (\ref{CNN-1}) and (\ref{QNN-1}), see Theorems
\ref{Theorem-2A} and \ref{Theorem-2B}.

\section{Basic concepts of $p$-adic analysis}

In this section, we fix the notation and collect some basic results on
$p$-adic analysis that we use in this paper. For a detailed exposition on
$p$-adic analysis, the reader may consult \cite{Zuniga-Textbook},
\cite{Kochubei}, \cite{V-V-Z}-\cite{Taibleson}.

\subsection{$p$-Adic numbers}

From now on, we use $p$ to denote a fixed prime number. Any non-zero $p$-adic
number $x$ has a unique expansion of the form%
\begin{equation}
x=x_{-k}p^{-k}+x_{-k+1}p^{-k+1}+\ldots+x_{0}+x_{1}p+\ldots,\text{ }
\label{p-adic-number}%
\end{equation}
with $x_{-k}\neq0$, where $k$ is an integer, and the $x_{j}$s\ are numbers
from the set $\left\{  0,1,\ldots,p-1\right\}  $. The set of all possible
sequences of the form (\ref{p-adic-number}) constitutes the field of $p$-adic
numbers $\mathbb{Q}_{p}$. There are natural field operations, sum and
multiplication, on series of form (\ref{p-adic-number}). There is also a norm
in $\mathbb{Q}_{p}$ defined as $\left\vert x\right\vert _{p}=p^{-ord(x)}$,
where $ord_{p}(x)=ord(x)=-k$, for a nonzero $p$-adic number $x$. By definition
$ord(0)=\infty$. The field of $p$-adic numbers with the distance induced by
$\left\vert \cdot\right\vert _{p}$ is a complete ultrametric space. The
ultrametric property refers to the fact that $\left\vert x-y\right\vert
_{p}\leq\max\left\{  \left\vert x-z\right\vert _{p},\left\vert z-y\right\vert
_{p}\right\}  $ for any $x$, $y$, $z$ in $\mathbb{Q}_{p}$. The $p$-adic
integers, which are sequences of form (\ref{p-adic-number}) with $-k\geq0$,
constitute the unit ball $\mathbb{Z}_{p}$. The unit ball is an infinite rooted
tree with fractal structure. As a topological space $\mathbb{Q}_{p}$\ is
homeomorphic to a Cantor-like subset of the real line, see, e.g.,
\cite{V-V-Z}-\cite{Alberio et al}.

\subsection{Test functions}

A function $\varphi:\mathbb{Q}_{p}\rightarrow\mathbb{C}$ is called locally
constant, if for any $a\in\mathbb{Q}_{p}$, there is an integer $l=l(a)$, such
that
\[
\varphi\left(  a+x\right)  =\varphi\left(  a\right)  \text{ for any }%
|x|_{p}\leq p^{l}.
\]
The set of functions for which $l=l\left(  \varphi\right)  $ depends only on
$\varphi$ form a $\mathbb{C}$-vector space denoted as $\mathcal{U}%
_{loc}\left(  \mathbb{Q}_{p}\right)  $. We call $l\left(  \varphi\right)  $
the exponent of local constancy. If $\varphi\in\mathcal{U}_{loc}\left(
\mathbb{Q}_{p}\right)  $ has compact support, we say that $\varphi$ is a test
function. We denote by $\mathcal{D}(\mathbb{Q}_{p})$ the complex vector space
of test functions (or Bruhat-Schwartz functions). There is a natural
integration theory so that $\int_{\mathbb{Q}_{p}}\varphi\left(  x\right)  dx$
gives a well-defined complex number. The measure $dx$ is the Haar measure of
$\mathbb{Q}_{p}$, \cite{Halmos}.

\subsection{Lebesgue spaces}

Let $U\subset\mathbb{Q}_{p}$ be an open subset. We denote by $\mathcal{D}(U)$ the $\mathbb{C}$-vector space of test functions with supports
contained in $U$; then, $\mathcal{D}(U)$ is dense in
\[
L^{\rho}(U)=\left\{  \varphi:U\rightarrow\mathbb{C};\left\Vert \varphi
\right\Vert _{\rho}=\left\{
{\displaystyle\int\limits_{U}}
\left\vert \varphi\left(  x\right)  \right\vert ^{\rho}dx\right\}  ^{\frac
{1}{\rho}}<\infty\right\}  ,
\]
for $1\leq\rho<\infty$, see, e.g., \cite[Section 4.3]{Alberio et al}.

\section{Existence of Bumps in $p$-adic QNNs}

Only in this section, we set 
\[
J\left( \left\vert x\right\vert _{p}\right) :\mathbb{Z}_{p}\rightarrow 
\mathbb{R}_{\geq 0},\left\Vert J\right\Vert _{1}=\int_{\mathbb{Z}%
_{p}}J\left( \left\vert x\right\vert _{p}\right) dx\in \left( 0,1\right) ,
\]%
$W\left( \left\vert x\right\vert _{p}\right) $, $Z(\left\vert x\right\vert
_{p})\in L^{1}\left( \mathbb{Z}_{p}\right) $.  We assume that $\phi :\mathbb{%
R}\rightarrow \mathbb{R}$ is Lipschitz, i.e., there exists a positive
constant $L_{\phi }$ such that $\left\vert \phi \left( a\right) -\phi \left(
b\right) \right\vert \leq L_{\phi }\left\vert a-b\right\vert $, for any $a$, 
$b\in \mathbb{R}$. We extend $\phi :\mathbb{C}\rightarrow \mathbb{C}$, by
taking $\phi \left( a+ib\right) =\phi \left( a\right) +i\phi \left( b\right) 
$, for $a$, $b\in \mathbb{R}$. Relevant examples of activation functions are 
$\phi \left( s\right) =\tanh (s)$, and $\varphi (s)=\frac{1}{2}\left(
|s+1|-|s-1|\right) .$

We consider a (convolutional) quantum network of the form%
\begin{equation}
i\frac{\partial}{\partial t}\Psi\left(  x,t\right)  =-J\left(  \left\vert
x\right\vert _{p}\right)  \ast\Psi\left(  x,t\right)  +\Psi\left(  x,t\right)
+%
{\displaystyle\int\limits_{\mathbb{Z}_{p}}}
W\left(  \left\vert x-y\right\vert _{p}\right)  \phi\left(  \Psi\left(
y,t\right)  \right)  dy+Z(\left\vert x\right\vert _{p}).
\label{EQ_SCHRODINGER_4}%
\end{equation}

In this section, we show the existence of time-independent solutions (bumps)
for (\ref{EQ_SCHRODINGER_4}).

\begin{theorem}
\label{Theorem-2A}Assume that $J,Z,W\in L^{1}\left(  \mathbb{Z}_{p}\right)  $, with $\left\Vert J\right\Vert _{1}\in \left( 0,1\right) $, 
and $\left\Vert J\right\Vert _{1}+L_{\phi}\left\Vert W\right\Vert _{1}%
<1$. Then, the integral equation%
\begin{equation}
\psi\left(  x\right)  =J\left(  \left\vert x\right\vert _{p}\right)  \ast
\psi\left(  x\right)  -%
{\displaystyle\int\limits_{\mathbb{Z}_{p}}}
W\left(  \left\vert x-y\right\vert \right)  \phi\left(  \psi\left(  y\right)
\right)  dy-Z(x) \label{EQ_10}%
\end{equation}
has a unique solution $\psi\in L^{1}\left(  \mathbb{Z}_{p}\right)  $.
\end{theorem}

\begin{proof}
Take $f,g\in L^{1}\left(  \mathbb{Z}_{p}\right)  $, and set%
\[
\left(  Tf\right)  \left(  x\right)  =:J\left(  \left\vert x\right\vert
_{p}\right)  \ast f\left(  x\right)  -%
{\displaystyle\int\limits_{\mathbb{Z}_{p}}}
W\left(  \left\vert x-y\right\vert \right)  \phi\left(  f\left(  y\right)
\right)  dy-Z(x).
\]
Then
\[
\left\Vert \left(  Tf-Tg\right)  \left(  x\right)  \right\Vert _{1}%
\leq\left\Vert J\ast\left(  f-g\right)  \right\Vert _{1}+\left\Vert \text{ }%
{\displaystyle\int\limits_{\mathbb{Z}_{p}}}
W\left(  \cdot-y\right)  \left\{  \phi\left(  f\left(  y\right)  \right)
-\phi\left(  g\left(  y\right)  \right)  \right\}  dy\right\Vert _{1}.
\]
Now, $\left\Vert J\ast\left(  f-g\right)  \right\Vert _{1}\leq\left\Vert
J\right\Vert _{1}\left\Vert \left(  f-g\right)  \right\Vert _{1}$, cf.
\cite[Chapter II, Theorem 1.7]{Taibleson}, and using the fact that \ $\phi
$\ is a Lipschitz function, and Fubini's theorem,
\begin{gather*}
\left\Vert \text{ }%
{\displaystyle\int\limits_{\mathbb{Z}_{p}}}
W\left(  \left\vert \cdot-y\right\vert _{p}\right)  \left\{  \phi\left(
f\left(  y\right)  \right)  -\phi\left(  g\left(  y\right)  \right)  \right\}
dy\right\Vert _{1}=\\%
{\displaystyle\int\limits_{\mathbb{Z}_{p}}}
\text{ }\left\vert \text{ }%
{\displaystyle\int\limits_{\mathbb{Z}_{p}}}
W\left(  \left\vert x-y\right\vert _{p}\right)  \left\{  \phi\left(  f\left(
y\right)  \right)  -\phi\left(  g\left(  y\right)  \right)  \right\}
dy\right\vert dx\leq\\
L_{\phi}%
{\displaystyle\int\limits_{\mathbb{Z}_{p}}}
\text{ }%
{\displaystyle\int\limits_{\mathbb{Z}_{p}}}
\text{ }\left\vert W\left(  \left\vert x-y\right\vert _{p}\right)  \right\vert
\left\vert f(y)-g(y)\right\vert dydx=\\
L_{\phi}%
{\displaystyle\int\limits_{\mathbb{Z}_{p}}}
\text{ }\left\vert f(y)-g(y)\right\vert \left\{
{\displaystyle\int\limits_{\mathbb{Z}_{p}}}
\text{ }\left\vert W\left(  \left\vert x-y\right\vert _{p}\right)  \right\vert
dx\right\}  dy=\\
L_{\phi}%
{\displaystyle\int\limits_{\mathbb{Z}_{p}}}
\text{ }\left\vert f(y)-g(y)\right\vert \left\{
{\displaystyle\int\limits_{\mathbb{Z}_{p}}}
\text{ }\left\vert W\left(  \left\vert z\right\vert _{p}\right)  \right\vert
dz\right\}  dy=L_{\phi}\left\Vert W\right\Vert _{1}\left\Vert f-g\right\Vert
_{1}.
\end{gather*}
Consequently,%
\begin{equation}
\left\Vert \left(  Tf-Tg\right)  \left(  x\right)  \right\Vert _{1}\leq\left(
\left\Vert J\right\Vert _{1}+L_{\phi}\left\Vert W\right\Vert _{1}\right)
\left\Vert \left(  f-g\right)  \right\Vert _{1}. \label{Inequality_A}%
\end{equation}

Now, taking $g=0$ in (\ref{Inequality_A}),%
\begin{align*}
\left\Vert Tf\right\Vert _{1}  &  \leq\left\Vert \left(  Tf-T0\right)
\right\Vert _{1}+\left\Vert \left(  T0\right)  \right\Vert _{1}\\
&  \leq\left(  \left\Vert J\right\Vert _{1}+L_{\phi}\left\Vert W\right\Vert
_{1}\right)  \left\Vert f\right\Vert _{1}+\left\Vert \phi\left(  0\right)
{\displaystyle\int\limits_{\mathbb{Z}_{p}}}
W\left(  \left\vert \cdot-y\right\vert \right)  dy-Z\right\Vert _{1}\\
&  \leq\left(  \left\Vert J\right\Vert _{1}+L_{\phi}\left\Vert W\right\Vert
_{1}\right)  \left\Vert f\right\Vert _{1}+\left\Vert Z\right\Vert
_{1}+\left\vert \phi\left(  0\right)  \right\vert \left\vert \text{ }%
{\displaystyle\int\limits_{\mathbb{Z}_{p}}}
W\left(  \left\vert z\right\vert \right)  dz\right\vert ,
\end{align*}
which implies that $T:L^{1}\left(  \mathbb{Z}_{p}\right)  \rightarrow
L^{1}\left(  \mathbb{Z}_{p}\right)  $. Now, the result follows by the Banach
fixed point theorem.
\end{proof}

The Theorem \ref{Theorem-2A} is also valid for $p$-adic QNNs of type
\[
\psi\left(  x\right)  =J\left(  \left\vert x\right\vert _{p}\right)  \ast
\psi\left(  x\right)  -\phi\left(
{\displaystyle\int\limits_{\mathbb{Z}_{p}}}
W\left(  x,y\right)  \psi\left(  y\right)  dy+Z(x)\right)  .
\]

\section{Existence of Bumps in $p$-adic CNNs}

In this section, we show the existence of time-independent solutions (bumps) for $p$-adic neural networks of type
\begin{equation}
\frac{\partial}{\partial t}U(x,t)=-\alpha U(x,t)+%
{\displaystyle\int\limits_{\mathbb{Z}_{p}}}
W\left(  \left\vert x-y\right\vert _{p}\right)  \phi\left(  U(y,t)\right)  dy,
\label{EQ_1}%
\end{equation}
where $x\in\mathbb{Q}_{p}$, $t\geq0$, $\alpha>0$, $W:\mathbb{R}_{\geq
0}\rightarrow\mathbb{R}$, $W\left(  \left\vert x\right\vert _{p}\right)  \in
L^{1}\left(  \mathbb{Z}_{p}\right)  $, and $\phi:\mathbb{R}\rightarrow
\mathbb{R}$ is a Lipschitz function as before. In this section, all the
functions are real-valued.

\begin{theorem}
\label{Theorem-2B}Assume that $\alpha>0$, $W\left(  \left\vert x\right\vert
_{p}\right)  \in L^{1}\left(  \mathbb{Z}_{p}\right)  $, and that
$\frac{L_{\phi}\left\Vert W\right\Vert _{1}}{\alpha}<1$. Then, the integral
equation%
\begin{equation}
U(x)=\frac{1}{\alpha}%
{\displaystyle\int\limits_{\mathbb{Z}_{p}}}
W\left(  \left\vert x-y\right\vert _{p}\right)  \phi\left(  U\left(  y\right)
\right)  dy \label{EQ_10B}%
\end{equation}
has a unique solution $U\in L^{1}\left(  \mathbb{Z}_{p}\right)  $.
\end{theorem}

\begin{proof}
The proof is completely similar to the one given for Theorem
\ref{Theorem-2A}. Take $f,g\in L^{1}\left(  \mathbb{Z}_{p}\right)  $, and set%
\[
\left(  Tf\right)  \left(  x\right)  =:\frac{1}{\alpha}%
{\displaystyle\int\limits_{\mathbb{Q}_{p}}}
W\left(  \left\vert x-y\right\vert _{p}\right)  \phi\left(  f\left(  y\right)
\right)  dy.
\]
Using the fact that \ $\phi$\ is a Lipschitz function, and Fubini's theorem,
\begin{equation}
\left\Vert \left(  Tf-Tg\right)  \left(  x\right)  \right\Vert _{1}=\frac
{1}{\alpha}\left\Vert \text{ }%
{\displaystyle\int\limits_{\mathbb{Z}_{p}}}
W\left(  \cdot-y\right)  \left\{  \phi\left(  f\left(  y\right)  \right)
-\phi\left(  g\left(  y\right)  \right)  \right\}  dy\right\Vert _{1}\leq
\frac{L_{\phi}\left\Vert W\right\Vert _{1}}{\alpha}\left\Vert f-g\right\Vert
_{1}. \label{Inequality}%
\end{equation}
Now, taking $g=0$ in (\ref{Inequality}),%
\[
\left\Vert Tf\right\Vert _{1}\leq\left\Vert \left(  Tf-T0\right)  \right\Vert
_{1}+\left\Vert \left(  T0\right)  \right\Vert _{1}\leq\frac{L_{\phi
}\left\Vert W\right\Vert _{1}}{\alpha}\left\Vert f\right\Vert _{1},
\]
which implies that $T:L^{1}\left(  \mathbb{Z}_{p}\right)  \rightarrow
L^{1}\left(  \mathbb{Z}_{p}\right)  $. Finally, the result follows from the
Banach fixed-point theorem.
\end{proof}

The Theorem \ref{Theorem-2B} is valid for $p$-adic CNNs of type%
\[
U(x)=\frac{1}{\alpha}\phi\left(  {\int\limits_{\mathbb{Z}_{p}}}w(|x-y|_{p}%
)U(y)\,dy+Z(x)\right)
\]
if $Z\in L^{1}(\mathbb{Z}_{p})$.

\section{Convolution operators in the $p$-adic unit ball}

We fix a continuous function $F:\left[  0,1\right]  \rightarrow\mathbb{C}$,
and define the radial complex-valued function $F(\left\vert x\right\vert
_{p})$, for $x\in\mathbb{Z}_{p}$. Notice that $F(\left\vert x\right\vert
_{p})$ is a continuous function on the unit ball, and consequently
$F(\left\vert x\right\vert _{p})\in L^{1}\left(  \mathbb{Z}_{p}\right)  $. We
now define the operator
\[
\left(  \boldsymbol{F}g\right)  \left(  x\right)  =F(\left\vert x\right\vert
_{p})\ast g(x)\text{, for }g(x)\in L^{\rho}\left(  \mathbb{Z}_{p}\right)
\text{, }\rho\in\left[  1,\infty\right]  ,
\]
since $\left\Vert F(\left\vert x\right\vert _{p})\ast g(x)\right\Vert _{\rho
}\leq\left\Vert F(\left\vert x\right\vert _{p})\right\Vert _{1}\left\Vert
g(x)\right\Vert _{\rho}$, and thus the mapping
\[
\boldsymbol{F}:L^{\rho}\left(  \mathbb{Z}_{p}\right)  \rightarrow L^{\rho
}\left(  \mathbb{Z}_{p}\right)
\]
is a bounded linear operator.

We now study the functions $\boldsymbol{F}\left(  1_{S_{-j}}(x)\right)  $, and
$\boldsymbol{F}\left(  \Omega\left(  p^{M+1}\left\vert x\right\vert _{p}\right)  \right)
$, where $j$, $M$ are natural numbers. Notice that%
\[
F(\left\vert x\right\vert _{p})\ast1_{S_{-j}}(x)={\displaystyle\int
\limits_{\mathbb{Z}_{p}}}1_{S_{-j}}(y)F(\left\vert x-y\right\vert _{p})dy=%
{\displaystyle\int\limits_{p^{j}\mathbb{Z}_{p}^{\times}}}
F(\left\vert x-y\right\vert _{p})dy,
\]
where%
\[
\mathbb{Z}_{p}^{\times}=\left\{  x\in\mathbb{Z}_{p};\text{ }\left\vert
x\right\vert _{p}=1\right\}  .
\]

\begin{lemma}
\label{Lemma_A1}Set $I_{j}(x):=%
{\displaystyle\int\nolimits_{p^{j}\mathbb{Z}_{p}^{\times}}}
F(\left\vert x-y\right\vert _{p})dy$, for $x\in\mathbb{Z}_{p}$, $j\in
\mathbb{N}$. Then, the following formulas hold.

\noindent(i) If $\left\vert x\right\vert _{p}>p^{-j}$, then
\[
I_{j}(x)=\left(  1-p^{-1}\right)  p^{-j}F(\left\vert x\right\vert _{p})\text{
if }j\geq1.
\]
The \ case $j=0$ is impossible since $\left\vert x\right\vert _{p}\leq1$.

\noindent(ii) If $\left\vert x\right\vert _{p}<p^{-j}$, then%
\[
I_{j}(x)=%
{\displaystyle\int\limits_{p^{j}\mathbb{Z}_{p}^{\times}}}
F(\left\vert y\right\vert _{p})dy=p^{-j}\left(  1-p^{-1}\right)  F(p^{-j}).
\]

\noindent(iii) If $\left\vert x\right\vert _{p}=p^{-j}$, then%
\[
I_{j}(x)=\left(  1-2p^{-1}\right)  F\left(  p^{-j}\right)  +\left(
1-p^{-1}\right)
{\displaystyle\sum\limits_{k=1}^{\infty}}
p^{-k}F\left(  p^{-j-k}\right)  .
\]

\end{lemma}

\begin{proof}
Parts (i) and (ii) follow directly from the ultrametric property of
$\left\vert \cdot\right\vert _{p}$: if $\left\vert y\right\vert _{p}=p^{-j}$,
and $\left\vert x\right\vert _{p}\neq p^{-j}$, then%
\[
\left\vert x-y\right\vert _{p}=\max\left\{  \left\vert x\right\vert
_{p},p^{-j}\right\}  .
\]
To establish the last formula, we use that $x=p^{j}u$, with $u\in
\mathbb{Z}_{p}^{\times}$, then
\begin{align*}
I_{j}(x)  &  =%
{\displaystyle\int\limits_{p^{j}\mathbb{Z}_{p}^{\times}}}
F(\left\vert p^{j}u-y\right\vert _{p})dy=%
{\displaystyle\int\limits_{\mathbb{Z}_{p}^{\times}}}
F(\left\vert p^{j}\left(  u-z\right)  \right\vert _{p})dy=%
{\displaystyle\sum\limits_{k=0}^{\infty}}
\text{ \ }%
{\displaystyle\int\limits_{\substack{\left\vert u-z\right\vert _{p}%
=p^{-k}\\z\in\mathbb{Z}_{p}^{\times}}}}
F(\left\vert p^{j}\left(  u-z\right)  \right\vert _{p})dy\\
&  =%
{\displaystyle\int\limits_{\mathbb{Z}_{p}^{\times}\smallsetminus\left(
u+\mathbb{Z}_{p}^{\times}\right)  }}
F(p^{-j})dy+%
{\displaystyle\sum\limits_{k=1}^{\infty}}
\text{ \ }%
{\displaystyle\int\limits_{\substack{\left\vert u-z\right\vert _{p}%
=p^{-k}\\z\in\mathbb{Z}_{p}^{\times}}}}
F(\left\vert p^{j}\left(  u-z\right)  \right\vert _{p})dy\\
&  =\left(  1-2p^{-1}\right)  F(p^{-j})+%
{\displaystyle\sum\limits_{k=1}^{\infty}}
\text{ \ }F(p^{-j-k})%
{\displaystyle\int\limits_{u+p^{-k}\mathbb{Z}_{p}^{\times}}}
dy=\left(  1-p^{-1}\right)
{\displaystyle\sum\limits_{k=0}^{\infty}}
p^{-k}F(p^{-j-k}),
\end{align*}
were we used \ that
\[%
{\displaystyle\int\limits_{\mathbb{Z}_{p}^{\times}\smallsetminus\left(
u+\mathbb{Z}_{p}^{\times}\right)  }}
dy=%
{\displaystyle\sum\limits_{\substack{l=1\\l\neq l_{0}}}^{p-1}}
\text{ \ }%
{\displaystyle\int\limits_{l+p\mathbb{Z}_{p}}}
dy=\left(  p-2\right)  p^{-1}.
\]

\end{proof}

\begin{lemma}
\label{Lemma_A2}With\ the above notation, and with $M$ a positive integer,%
\begin{align*}
F(\left\vert x\right\vert _{p})\ast\Omega\left(  p^{M+1}\left\vert
x\right\vert _{p}\right)   &  =\left(  \text{ }%
{\displaystyle\int\limits_{p^{M+1}\mathbb{Z}_{p}^{\times}}}
F(\left\vert y\right\vert _{p})dy\right)  \Omega\left(  p^{M+1}\left\vert
x\right\vert _{p}\right)  +\\
&  p^{-M-1}%
{\displaystyle\sum\limits_{j=0}^{M-1}}
F(p^{-j})1_{S_{-j}}(x).
\end{align*}

\end{lemma}

\begin{proof}
The announced formula follows from%
\begin{align*}
F(\left\vert x\right\vert _{p})\ast\Omega\left(  p^{M+1}\left\vert
x\right\vert _{p}\right)   &  =%
{\displaystyle\int\limits_{\mathbb{Z}_{p}}}
\Omega\left(  p^{M+1}\left\vert y\right\vert _{p}\right)  F(\left\vert
x-y\right\vert _{p})dy=\\%
{\displaystyle\int\limits_{p^{M+1}\mathbb{Z}_{p}}}
F(\left\vert x-y\right\vert _{p})dy  &  =\left\{
\begin{array}
[c]{ll}%
{\displaystyle\int\limits_{p^{M+1}\mathbb{Z}_{p}}}
F(\left\vert y\right\vert _{p})dy & \text{if }x\in p^{M+1}\mathbb{Z}_{p}\\
& \\
F(\left\vert x\right\vert _{p})p^{-M-1} & \text{if }x\notin p^{M+1}%
\mathbb{Z}_{p},
\end{array}
\right.
\end{align*}
by using that, if $x\notin p^{M+1}\mathbb{Z}_{p}$,
\[
F(\left\vert x\right\vert _{p})=%
{\displaystyle\sum\limits_{j=0}^{M-1}}
F(p^{-j})1_{S_{-j}}(x).
\]

\end{proof}

For $j\in\mathbb{N}=\left\{  0,1,\ldots,n,n+1,\ldots\right\}  $, we define the
sphere with center at the origin and radius $p^{-j}$ as%
\[
S_{-j}=\left\{  x\in\mathbb{Z}_{p};\text{ }\left\vert x\right\vert _{p}%
=p^{-j}\right\}  .
\]

Now, for a positive integer $M$, the set of test functions%
\begin{equation}
\left\{  1_{S_{0}}(x),1_{S_{-1}}(x),\ldots,1_{S_{-M}}(x),\Omega\left(
p^{M+1}\left\vert x\right\vert _{p}\right)  \right\}  \label{Basis}%
\end{equation}
is linearly independent, because the supports of any of two of these functions
are disjoint. We denote by $\mathcal{R}_{M}(\mathbb{Z}_{p})$, the $\mathbb{C}%
$-vector spanned by this set.

\begin{lemma}
\label{Lemma_A3}With the above notation,%
\begin{gather*}
F(\left\vert x\right\vert _{p})\ast1_{S_{0}}(x)=\left\{  \left(
1-2p^{-1}\right)  F(1)+\left(  1-p^{-1}\right)
{\displaystyle\sum\limits_{k=1}^{\infty}}
p^{-k}F\left(  p^{-k}\right)  \right\}  1_{S_{0}}(x)+\\
\left(  \text{ }%
{\displaystyle\int\limits_{\mathbb{Z}_{p}^{\times}}}
F(\left\vert y\right\vert _{p})dy\right)
{\displaystyle\sum\limits_{j=1}^{M}}
1_{S_{-j}}(x)+\left(  \text{ }%
{\displaystyle\int\limits_{\mathbb{Z}_{p}^{\times}}}
F(\left\vert y\right\vert _{p})dy\right)  \Omega\left(  p^{M+1}\left\vert
x\right\vert _{p}\right)  .
\end{gather*}

\end{lemma}

\begin{proof}
Using that%
\[
\Omega\left(  p\left\vert x\right\vert _{p}\right)  =\Omega\left(
p^{M+1}\left\vert x\right\vert _{p}\right)  +%
{\displaystyle\sum\limits_{j=1}^{M}}
1_{S_{-j}}(x),
\]
and Lemma \ref{Lemma_A1}, parts (ii)-(iii), and
\[%
{\displaystyle\int\limits_{\mathbb{Z}_{p}^{\times}}}
F(\left\vert y\right\vert _{p})dy=\left(  1-p^{-1}\right)  F(1),
\]
we have%
\begin{gather*}
F(\left\vert x\right\vert _{p})\ast1_{S_{0}}(x)=\left(  \text{ }%
{\displaystyle\int\limits_{\mathbb{Z}_{p}^{\times}}}
F(\left\vert y\right\vert _{p})dy\right)  \Omega\left(  p\left\vert
x\right\vert _{p}\right)  +\\
\left\{  \left(  1-2p^{-1}\right)  F(1)+\left(  1-p^{-1}\right)
{\displaystyle\sum\limits_{k=1}^{\infty}}
p^{-k}F\left(  p^{-k}\right)  \right\}  1_{S_{0}}(x)=\\
\left\{  \left(  1-2p^{-1}\right)  F(1)+\left(  1-p^{-1}\right)
{\displaystyle\sum\limits_{k=1}^{\infty}}
p^{-k}F\left(  p^{-k}\right)  \right\}  1_{S_{0}}(x)+\left(  \text{ }%
{\displaystyle\int\limits_{\mathbb{Z}_{p}^{\times}}}
F(\left\vert y\right\vert _{p})dy\right)
{\displaystyle\sum\limits_{j=1}^{M}}
1_{S_{-j}}(x)+\\
\left(  \text{ }%
{\displaystyle\int\limits_{\mathbb{Z}_{p}^{\times}}}
F(\left\vert y\right\vert _{p})dy\right)  \Omega\left(  p^{M+1}\left\vert
x\right\vert _{p}\right)  .
\end{gather*}

\end{proof}

\begin{lemma}\label{Lemma_A4}
With the above notation, for $j=1,\ldots,M$,%
\begin{gather*}
F(\left\vert x\right\vert _{p})\ast1_{S_{-j}}(x)=\left(  1-p^{-1}\right)
p^{-j}{\displaystyle\sum\limits_{k=0}^{j-1}}F(p^{-k})1_{S_{-k}}(x)+\\
\left\{  \left(  1-2p^{-1}\right)  F\left(  p^{-j}\right)  +\left(
1-p^{-1}\right)  {\displaystyle\sum\limits_{k=1}^{\infty}}p^{-k}F\left(
p^{-j-k}\right)  \right\}  1_{S_{-j}}(x)+\\
\left(  1-p^{-1}\right)  p^{-j}F(p^{-j}){\displaystyle\sum\limits_{k=j+1}^{M}%
}1_{S_{-k}}(x)+\left(  1-p^{-1}\right)  p^{-j}F(p^{-j})\Omega\left(
p^{M+1}\left\vert x\right\vert _{p}\right)  ,
\end{gather*}
with the convention that
\[
{\displaystyle\sum\limits_{k=M+1}^{M}}1_{S_{-k}}(x)=0.
\]

\end{lemma}

\begin{proof}
Using that%
\[
\Omega\left(  \left\vert x\right\vert _{p}\right)  =\Omega\left(
p^{j+1}\left\vert x\right\vert _{p}\right)  +%
{\displaystyle\sum\limits_{k=0}^{j}}
1_{S_{-k}}(x),
\]%
\[
F(\left\vert x\right\vert _{p})\ast1_{S_{-j}}(x)=%
{\displaystyle\sum\limits_{k=0}^{j}}
I_{j}(x)1_{S_{-k}}(x)+I_{j}(x)\Omega\left(  p^{j+1}\left\vert x\right\vert
_{p}\right)  .
\]
We now use
\[
\Omega\left(  p^{j+1}\left\vert x\right\vert _{p}\right)  =\Omega\left(
p^{M+1}\left\vert x\right\vert _{p}\right)  +%
{\displaystyle\sum\limits_{k=j+1}^{M}}
1_{S_{-k}}(x),
\]
Lemma \ref{Lemma_A1}, parts (i)-(ii)-(iii) , and
\[
\left(  \text{ }{\displaystyle\int\limits_{p^{j}\mathbb{Z}_{p}^{\times}}%
}F(\left\vert y\right\vert _{p})dy\right)  =\left(  1-p^{-1}\right)
p^{-j}F(p^{-j}),
\]
to get%
\begin{gather*}
F(\left\vert x\right\vert _{p})\ast1_{S_{-j}}(x)=\left(  1-p^{-1}\right)
p^{-j}%
{\displaystyle\sum\limits_{k=0}^{j-1}}
F(p^{-k})1_{S_{-k}}(x)+\\
\left\{  \left(  1-2p^{-1}\right)  F\left(  p^{-j}\right)  +\left(
1-p^{-1}\right)
{\displaystyle\sum\limits_{k=1}^{\infty}}
p^{-k}F\left(  p^{-j-k}\right)  \right\}  1_{S_{-j}}(x)\\
+\Omega\left(  p^{j+1}\left\vert x\right\vert _{p}\right)
{\displaystyle\int\limits_{p^{j}\mathbb{Z}_{p}^{\times}}}
F(\left\vert y\right\vert _{p})dy=\left(  1-p^{-1}\right)  p^{-j}%
{\displaystyle\sum\limits_{k=0}^{j-1}}
F(p^{-k})1_{S_{-k}}(x)\\
+\left\{  \left(  1-2p^{-1}\right)  F\left(  p^{-j}\right)  +\left(
1-p^{-1}\right)
{\displaystyle\sum\limits_{k=0}^{\infty}}
p^{-k}F\left(  p^{-j-k}\right)  \right\}  1_{S_{-j}}(x)\\
+\left(  1-p^{-1}\right)  p^{-j}F(p^{-j})%
{\displaystyle\sum\limits_{k=j+1}^{M}}
1_{S_{-k}}(x)+\left(  1-p^{-1}\right)  p^{-j}F(p^{-j})\Omega\left(
p^{M+1}\left\vert x\right\vert _{p}\right)  .
\end{gather*}

\end{proof}

We now define the matrix $\mathbb{A}_{\boldsymbol{F}}=\left[  A_{i,j}\right]
_{0\leq i,j\leq M+1}$, where the columns are given as follows:%
\begin{equation}
\left[  A_{i,0}\right]  _{0\leq i\leq M+1}=\left[
\begin{array}
[c]{c}%
\begin{array}
[c]{c}%
\left[  \left(  1-2p^{-1}\right)  F(1)+\left(  1-p^{-1}\right)
{\displaystyle\sum\limits_{k=1}^{\infty}}
p^{-k}F\left(  p^{-k}\right)  \right]  _{1\times1}\\
\\
\left[
\begin{array}
[c]{c}%
\left(  1-p^{-1}\right)  F(1)\\
\vdots\\
\left(  1-p^{-1}\right)  F(1)
\end{array}
\right]  _{\left(  M+1\right)  \times1}%
\end{array}
\end{array}
\right]  , \label{A_1}%
\end{equation}
for $j=1,\ldots,M-1$,%
\begin{equation}
\left[  A_{i,j}\right]  _{0\leq i\leq M+1}=\left[
\begin{array}
[c]{c}%
\begin{array}
[c]{c}%
\left[
\begin{array}
[c]{c}%
\left(  1-p^{-1}\right)  p^{-1}F(p^{-0})\\
\vdots\\
\left(  1-p^{-1}\right)  p^{-j}F(p^{-\left(  j-1\right)  })
\end{array}
\right]  _{j\times1}%
\end{array}
\\
\\
\left[  \left(  1-2p^{-1}\right)  F\left(  p^{-j}\right)  +\left(
1-p^{-1}\right)
{\displaystyle\sum\limits_{k=1}^{\infty}}
p^{-k}F\left(  p^{-j-k}\right)  \right]  _{1\times1}\\
\\
\text{ }\left[
\begin{array}
[c]{c}%
\left(  1-p^{-1}\right)  p^{-j}F(p^{-j})\\
\vdots\\
\left(  1-p^{-1}\right)  p^{-j}F(p^{-j})
\end{array}
\right]  _{\left(  M-j+1\right)  \times1}%
\end{array}
\right]  ; \label{A_2}%
\end{equation}
for $j=M$,%
\begin{equation}
\left[  A_{i,M}\right]  _{0\leq i\leq M+1}=\left[
\begin{array}
[c]{c}%
\begin{array}
[c]{c}%
\left[
\begin{array}
[c]{c}%
\left(  1-p^{-1}\right)  p^{-1}F(p^{-0})\\
\vdots\\
\left(  1-p^{-1}\right)  p^{-j}F(p^{-\left(  j-1\right)  })
\end{array}
\right]  _{j\times1}%
\end{array}
\\
\\
\left[  \left(  1-2p^{-1}\right)  F\left(  p^{-j}\right)  +\left(
1-p^{-1}\right)
{\displaystyle\sum\limits_{k=1}^{\infty}}
p^{-k}F\left(  p^{-j-k}\right)  \right]  _{1\times1}\\
\\
\text{ }\left[
\begin{array}
[c]{c}%
0\\
\vdots\\
0\\
\left(  1-p^{-1}\right)  p^{-j}F(p^{-j})
\end{array}
\right]  _{\left(  M-j+1\right)  \times1}%
\end{array}
\right]  ; \label{A_3}%
\end{equation}
for $j=M+1$,%
\begin{equation}
\left[  A_{i,M+1}\right]  _{0\leq i\leq M+1}=\left[
\begin{array}
[c]{c}%
p^{-M-1}F(p^{-0})\\
p^{-M-1}F(p^{-1})\\
\vdots\\
p^{-M-1}F(p^{-j})\\
\vdots\\
p^{-M}F(p^{-(M)})\\
\\
\left(  1-p^{-1}\right)  p^{-M-1}F(p^{-M-1})
\end{array}
\right]  . \label{A_4}%
\end{equation}

As a consequence of the Lemmas \ref{Lemma_A1}, \ref{Lemma_A2}, \ref{Lemma_A3},
\ and \ref{Lemma_A4}, we have the following result.

\begin{theorem}
\label{Theorem_A}The $\mathbb{C}$-vector space $\mathcal{R}_{M}(\mathbb{Z}%
_{p})$ is invariant under operator $\boldsymbol{F}$, more precisely,
\[
\boldsymbol{F}:\mathcal{R}_{M}(\mathbb{Z}_{p})\rightarrow\mathcal{R}%
_{M}(\mathbb{Z}_{p})
\]
is a linear bounded operator. In the basis (\ref{Basis}), $\boldsymbol{F}$
corresponds to the operator on $\mathbb{R}^{M+2}$ induced by the matrix
$\mathbb{A}_{\boldsymbol{F}}=$ $\left[  A_{i,j}\right]  _{0\leq i,j\leq M+1}$.
Furthermore,%
\[
\left\Vert \mathbb{A}_{\boldsymbol{F}}\right\Vert =\max_{0\leq i,j\leq
M+1}\left\vert A_{i,j}\right\vert \leq\left\Vert F\right\Vert _{\infty}.
\]

\end{theorem}

\section{Pseudo-traveling waves for $p$-adic Schr\"{o}dinger equations with
non-local potentials}

We now apply the results of the previous section to approximate traveling
waves of the Schr\"{o}dinger equation (\ref{EQ_SCHRODINGER_4}), under the
hypotheses that $J\left(  \left\vert x\right\vert _{p}\right)  $, $W\left(
\left\vert x\right\vert _{p}\right)  $, $Z(\left\vert x\right\vert _{p})$ are
radial continuous functions defined on $\mathbb{Z}_{p}$,\ with $\int
_{\mathbb{Z}_{p}}J\left(  \left\vert x\right\vert _{p}\right)  dx=1$, and that
$\phi$ is Lipschitz.

We are interested in traveling waves, i.e., functions of type, $\Psi\left(
\left\vert x\right\vert _{p}+tv\right)  $, where $\Psi:\left[  0,\infty
\right)  \rightarrow\mathbb{C}$, $x\in\mathbb{Z}_{p}$, $t\geq0$, and $v>0$.
Notice that
\begin{align*}
\Psi\left(  \left\vert x\right\vert _{p}+tv\right)   &  =%
{\displaystyle\sum\limits_{j=0}^{\infty}}
\Psi\left(  p^{-j}+tv\right)  1_{S_{-j}}(x)\\
&  =%
{\displaystyle\sum\limits_{j=0}^{M}}
\Psi\left(  p^{-j}+tv\right)  1_{S_{-j}}(x)+\Omega\left(  p^{M+1}\left\vert
x\right\vert _{p}\right)  \Psi\left(  \left\vert x\right\vert _{p}+tv\right)
,
\end{align*}
for an arbitrary positive integer $M$. If $M$ is sufficiently large, and
$\Psi\left(  \left\vert x\right\vert _{p}+tv\right)  $ is continuous, \ then%
\[
\Omega\left(  p^{M+1}\left\vert x\right\vert _{p}\right)  \Psi\left(
\left\vert x\right\vert _{p}+tv\right)  \approx\Omega\left(  p^{M+1}\left\vert
x\right\vert _{p}\right)  \Psi\left(  tv\right)  .
\]
For this reason, solutions of the form (pseudo-traveling waves)%
\begin{equation}
\Psi\left(  x,t\right)  =%
{\displaystyle\sum\limits_{j=0}^{M}}
\Psi_{j}\left(  p^{-j}+tv\right)  1_{S_{-j}}(x)+\Omega\left(  p^{M+1}%
\left\vert x\right\vert _{p}\right)  \Psi_{M+1}\left(  tv\right)
\in\mathcal{R}_{M}(\mathbb{Z}_{p}).\label{Radial_solution}%
\end{equation}
are good approximations for traveling waves, under the hypothesis that
$\Psi\left(  \left\vert x\right\vert _{p}+tv\right)  $ is continuous. If this
hypothesis is not true, we show in this section that pseudo-traveling waves
are solutions of (\ref{EQ_SCHRODINGER_4}).

Since, the space $\mathcal{R}_{M}(\mathbb{Z}_{p})$ is invariant under
convolution operators, more precisely under $f(x)\rightarrow\left(
\boldsymbol{J}f\right)  \left(  x\right)  =J\left(  \left\vert x\right\vert
_{p}\right)  \ast f(x)$, $f(x)\rightarrow\left(  \boldsymbol{W}f\right)
\left(  x\right)  =W\left(  \left\vert x\right\vert _{p}\right)  \ast f(x)$,
then, the Schr\"{o}dinger equation (\ref{EQ_SCHRODINGER_4}), has a natural
discretization to $\mathcal{R}_{M}(\mathbb{Z}_{p})$, obtained by replacing
(\ref{Radial_solution}) in (\ref{EQ_SCHRODINGER_4}), under the assumption that%
\begin{equation}
Z(\left\vert x\right\vert _{p})=%
{\displaystyle\sum\limits_{j=0}^{M}}
Z\left(  p^{-j}\right)  1_{S_{-j}}(x)+Z(0)\Omega\left(  p^{M+1}\left\vert
x\right\vert _{p}\right)  \in\mathcal{R}_{M}(\mathbb{Z}_{p}).
\label{Zeta_Function}%
\end{equation}
The discretization of\ (\ref{EQ_SCHRODINGER_4}) is%
\begin{gather*}
i\frac{\partial}{\partial t}\left[
\begin{array}
[c]{c}%
\Psi_{0}\left(  p^{-0}+tv\right) \\
\vdots\\
\Psi_{M}\left(  p^{-M}+tv\right) \\
\Psi_{M+1}\left(  tv\right)
\end{array}
\right]  =\left(  \mathbb{I}-\mathbb{A}_{\boldsymbol{J}}\right)  \left[
\begin{array}
[c]{c}%
\Psi_{0}\left(  p^{-0}+tv\right) \\
\vdots\\
\Psi_{M}\left(  p^{-M}+tv\right) \\
\Psi_{M+1}\left(  tv\right)
\end{array}
\right]  +\\
\mathbb{W}_{\boldsymbol{J}}\left(  \left[
\begin{array}
[c]{c}%
\phi\left(  \Psi_{0}\left(  p^{-0}+tv\right)  \right) \\
\vdots\\
\phi\left(  \Psi_{M}\left(  p^{-M}+tv\right)  \right) \\
\phi\left(  \Psi_{M+1}\left(  tv\right)  \right)
\end{array}
\right]  \right)  +\left[
\begin{array}
[c]{c}%
Z\left(  p^{-0}\right) \\
\vdots\\
Z\left(  p^{-M}\right) \\
Z\left(  0\right)
\end{array}
\right]  ,
\end{gather*}
where the matrices $\mathbb{A}_{\boldsymbol{J}}$, $\mathbb{W}_{\boldsymbol{J}%
}$ are defined in Theorem \ref{Theorem_A}.

Notice that $\Psi_{j}\left(  p^{-j}+tv\right)  =\left(  \Psi_{j}\circ
T_{j}\right)  \left(  tv\right)  :=\rho_{j}\left(  tv\right)  $, where
$T_{j}\left(  tv\right)  =\left(  p^{-j}+tv\right)  $, for $j=0,\ldots,M$, and
$\Psi_{M+1}\left(  tv\right)  :=\rho_{M+1}\left(  tv\right)  $.

Now, \ take we consider the system of ODEs%
\begin{gather*}
iv\frac{\partial}{\partial s}\left[
\begin{array}
[c]{c}%
\rho_{0}\left(  s\right) \\
\vdots\\
\rho_{M}\left(  s\right) \\
\rho_{M+1}\left(  s\right)
\end{array}
\right]  =\left(  \mathbb{I}-\mathbb{A}_{\boldsymbol{J}}\right)  \left[
\begin{array}
[c]{c}%
\rho_{0}\left(  s\right) \\
\vdots\\
\rho_{M}\left(  s\right) \\
\beta_{M+1}\left(  s\right)
\end{array}
\right]  +\mathbb{W}_{\boldsymbol{J}}\left(  \left[
\begin{array}
[c]{c}%
\phi\left(  \rho_{0}\left(  s\right)  \right) \\
\vdots\\
\phi\left(  \sigma_{M}\left(  s\right)  \right) \\
\phi\left(  \rho_{M+1}\left(  s\right)  \right)
\end{array}
\right]  \right) \\
+\left[
\begin{array}
[c]{c}%
\beta_{0}\\
\vdots\\
\beta_{M}\\
\beta_{M+1}%
\end{array}
\right]  ,
\end{gather*}
where%
\[
\left[
\begin{array}
[c]{c}%
\rho_{0}\left(  s\right) \\
\vdots\\
\rho_{M}\left(  s\right) \\
\rho_{M+1}\left(  s\right)
\end{array}
\right]  \in\mathbb{C}^{M+2}\text{, and }\left[
\begin{array}
[c]{c}%
\beta_{0}\\
\vdots\\
\beta_{M}\\
\beta_{M+1}%
\end{array}
\right]  =\left[
\begin{array}
[c]{c}%
Z\left(  p^{-0}\right) \\
\vdots\\
Z\left(  p^{-M}\right) \\
Z\left(  0\right)
\end{array}
\right]  .
\]

Take%
\begin{align*}
\rho &  :=\left[
\begin{array}
[c]{c}%
\rho_{0}\\
\vdots\\
\rho_{M}\\
\rho_{M+1}%
\end{array}
\right]  \text{, }\sigma:=\left[
\begin{array}
[c]{c}%
\sigma_{0}\\
\vdots\\
\sigma_{M}\\
\sigma_{M+1}%
\end{array}
\right]  \text{, }\phi\left(  \rho\right)  :=\left[
\begin{array}
[c]{c}%
\phi\left(  \rho_{0}\right) \\
\vdots\\
\phi\left(  \rho_{M}\right) \\
\phi\left(  \rho_{M+1}\right)
\end{array}
\right] \\
\text{ }\phi\left(  \sigma\right)   &  =\left[
\begin{array}
[c]{c}%
\phi\left(  \sigma_{0}\right) \\
\vdots\\
\phi\left(  \sigma_{M}\right) \\
\phi\left(  \sigma_{M+1}\right)
\end{array}
\right]  \text{, }\beta=\left[
\begin{array}
[c]{c}%
\beta_{0}\\
\vdots\\
\beta_{M}\\
\beta_{M+1}%
\end{array}
\right]  \in\mathbb{C}^{M+2},
\end{align*}
and
\[
\mathbb{H}\left(  \rho\right)  =\left(  \mathbb{I}-\mathbb{A}_{\boldsymbol{J}%
}\right)  \rho+\mathbb{W}_{\boldsymbol{J}}\phi\left(  \rho\right)  +\beta.
\]
Then%
\begin{gather*}
\left\Vert \mathbb{H}\left(  \rho\right)  -\mathbb{H}\left(  \sigma\right)
\right\Vert \leq\left\Vert \left(  \mathbb{I}-\mathbb{A}_{\boldsymbol{J}%
}\right)  \left(  \rho-\sigma\right)  \right\Vert +\left\Vert \mathbb{W}%
_{\boldsymbol{J}}\left(  \phi\left(  \rho\right)  -\phi\left(  \sigma\right)
\right)  \right\Vert \leq\\
\left(  \left\Vert \mathbb{I}-\mathbb{A}_{\boldsymbol{J}}\right\Vert \right)
\left\Vert \rho-\sigma\right\Vert +\left\Vert \mathbb{W}_{\boldsymbol{J}%
}\right\Vert \left\Vert \phi\left(  \rho\right)  -\phi\left(  \sigma\right)
\right\Vert \leq\\
\left(  \left\Vert \mathbb{I}-\mathbb{A}_{\boldsymbol{J}}\right\Vert +L_{\phi
}\left\Vert \mathbb{W}_{\boldsymbol{J}}\right\Vert \right)  \left\Vert
\rho-\sigma\right\Vert =L_{\mathbb{H}}\left\Vert \rho-\sigma\right\Vert ,
\end{gather*}

which is $\mathbb{H}:\mathbb{C}^{M+2}\rightarrow\mathbb{C}^{M+2}$ is a
globally Lipschitz mapping.

The following result follows directly from the fundamental theorem for ODEs:

\begin{proposition}
If $L_{\mathbb{H}}\in\left(  0,v\right)  $, then, the initial value problem%
\begin{equation}
\left\{
\begin{array}
[c]{l}%
\rho\left(  s\right)  \in C^{1}\left(  \left[  0,\infty\right)  ,\mathbb{C}%
^{M+2}\right) \\
\\
i\frac{\partial}{\partial s}\rho\left(  s\right)  =\frac{1}{v}\mathbb{H}%
\left(  \rho\left(  s\right)  \right) \\
\\
\rho\left(  0\right)  =\omega\in\mathbb{C}^{M+2}%
\end{array}
\right.  \label{IVP_Equation}%
\end{equation}
has a unique solution, where $\omega$ is an arbitrary complex number.
\end{proposition}

\begin{theorem}
\label{Theorem-3A}Assume that $J\left(  \left\vert x\right\vert _{p}\right)
$, $W\left(  \left\vert x\right\vert _{p}\right)  $, are radial continuous
functions defined on $\mathbb{Z}_{p}$,\ with $\int_{\mathbb{Z}_{p}}J\left(
\left\vert x\right\vert _{p}\right)  dx=1$, $\phi$ is Lipschitz, and
$L_{\mathbb{H}}\in\left(  0,v\right)  $, with $v>0$. Furthermore, that
$Z(\left\vert x\right\vert _{p})\in\mathcal{R}_{M}(\mathbb{Z}_{p})$ as in
(\ref{Zeta_Function}). Then, the Schr\"{o}dinger equation
(\ref{EQ_SCHRODINGER_4}) admits a pseudo-traveling wave of the form%
\[
{\displaystyle\sum\limits_{j=0}^{M}}
\rho_{j}\left(  p^{-j}+tv\right)  1_{S_{-j}}(x)+\Omega\left(  p^{M+1}%
\left\vert x\right\vert _{p}\right)  \rho_{M+1}\left(  tv\right)  ,
\]
where $\rho$ is a solution of the initial valued problem (\ref{IVP_Equation}),
and $v$ is an arbitrary real number.
\end{theorem}

\begin{remark}
(i) The condition  $\int_{\mathbb{Z}_{p}}J\left( \left\vert x\right\vert
_{p}\right) dx=1$ is not required.

\noindent(ii) The above theorem is also valid for $p$-adic CNNs of the form( \ref{CNN-1}).
\end{remark}

If there exists a traveling wave $\Psi(\left\vert x\right\vert _{p}+tv)$
belonging to $L^{\rho}(\mathbb{Z}_{p})$, for some $\rho\in\left[
1,\infty\right]  $, and for $t\geq0$, it can be approximated by a continuous
radial function in $\left\vert x\right\vert _{p}$, that in turn can be well
approximated by pseudo-traveling waves. If the continuity condition is not
valid, then the pseudo-traveling waves are solutions of QCNNs of type
(\ref{EQ_SCHRODINGER_4}). Similar observations are valid for $p$-adic CNNs of
type (\ref{EQ_1}).

On the other hand, \ bump solutions of (\ref{EQ_10}) can be approximated as
\begin{equation}
\Psi\left(  x\right)  \approx%
{\displaystyle\sum\limits_{j=0}^{M}}
\Psi_{j}\left(  p^{-j}\right)  1_{S_{-j}}(x)+\Omega\left(  p^{M+1}\left\vert
x\right\vert _{p}\right)  \Psi_{M+1}\left(  0\right)  \in\mathcal{R}%
_{M}(\mathbb{Z}_{p}). \label{Approximation}%
\end{equation}
Indeed, by the Theorem \ref{Theorem-2A}, $\Psi\left(  x\right)  \in
L^{1}\left(  \mathbb{Z}_{p}\right)  $, since the continuous function are dense
in $L^{1}\left(  \mathbb{Z}_{p}\right)  $, to obtain a good approximation of
$\Psi\left(  x\right)  $, we may assume without loss of generality that
$\Psi\left(  x\right)  $ is continuous. For continuous functions, the
approximation (\ref{Approximation}) is valid. Similar observations are valid
for the bump solutions of (\ref{EQ_10B}).

\section{Numerical simulations}

In the numerical simulation, we pick the kernel $J(x)$ to be the Bessel
potential:%
\[
J_{\alpha }\left( x\right) =\frac{1-p^{-\alpha }}{1-p^{\alpha -1}}\left\{
\left\vert x\right\vert _{p}^{\alpha -1}-p^{\alpha -1}\right\} \Omega \left(
\left\vert x\right\vert _{p}\right) ,
\]%
where $\alpha \in \mathbb{R}\smallsetminus \left\{ 1\right\} $. This
function satisfies $\ J_{\alpha }\left( x\right) \geq 0$, $\left\Vert
J_{\alpha }\right\Vert _{1}=1$, see \cite[Lemma 5.2]{Taibleson}. The
parameter $\alpha $ is set to $1.5$ in all our simulations. The
activation function is fixed as 
\[
\phi \left( s\right) =\frac{1}{2}\left( \left\vert s+1\right\vert
+\left\vert s-1\right\vert \right) \text{, }s\in \mathbb{R}.
\]%
This activation function is widely used in cellular neural networks, see 
\cite{Chua-Tamas, Chua, Slavova, Zambrano-Zuniga-1, Zambrano-Zuniga-2,
Zuniga et al}.

The prime $p$ is fixed at $2$, and $l=10$.  The time interval used in
the simulations is $\left[ 0,200\right] $, with a step $\Delta t=0.05$. The
Parameter $M$ in the approximation (\ref{Radial_solution}) is fixed to $9$.
The labels in the vertical scale correspond to the terms $\left\vert \Psi
_{j}\left( p^{-j}+tv\right) \right\vert 1_{S_{-j}}(x)$, for $j=0,\ldots ,9$,
and $\Omega \left( p^{10}\left\vert x\right\vert _{p}\right) \left\vert \Psi
_{10}\left( tv\right) \right\vert $. \ The Initial datum used in all the
experiments is $\Omega (p^{2}|x|_{p})$.

Another important choice is the kernel $W(x,y)$, which determines the
strength of the connection between neurons $x$ and $y$. A natural way to fix 
$W(x,y)$ is to train the network to perform a particular task. This
machine-learning approach is not used here.  In the numerical simulations, we use two different kernel types. The first, $W(x,y)=0$, in this case, the neurons are
not connected. In the second type $W(x,y)=W_{0}$ constant, which means that 
\begin{equation}
\int\limits_{\mathbb{Z}_{p}}W(x,y)\phi \left( \Psi \left( y,t\right) \right)
dy=W_{0}\int\limits_{\mathbb{Z}_{p}}\phi \left( \Psi \left( y,t\right)
\right) dy.  \label{Term_W}
\end{equation}%
By interpreting $\phi \left( \Psi \left( y,t\right) \right) $ as the output
of the neuron located at position $y$ at the time $t$, the formula (\ref%
{Term_W}) gives a weighted average of the neurons' outputs at time $t$.

\begin{figure}[htbp]
    \centering
    \begin{subfigure}{0.48\textwidth}
        \includegraphics[width=\linewidth]{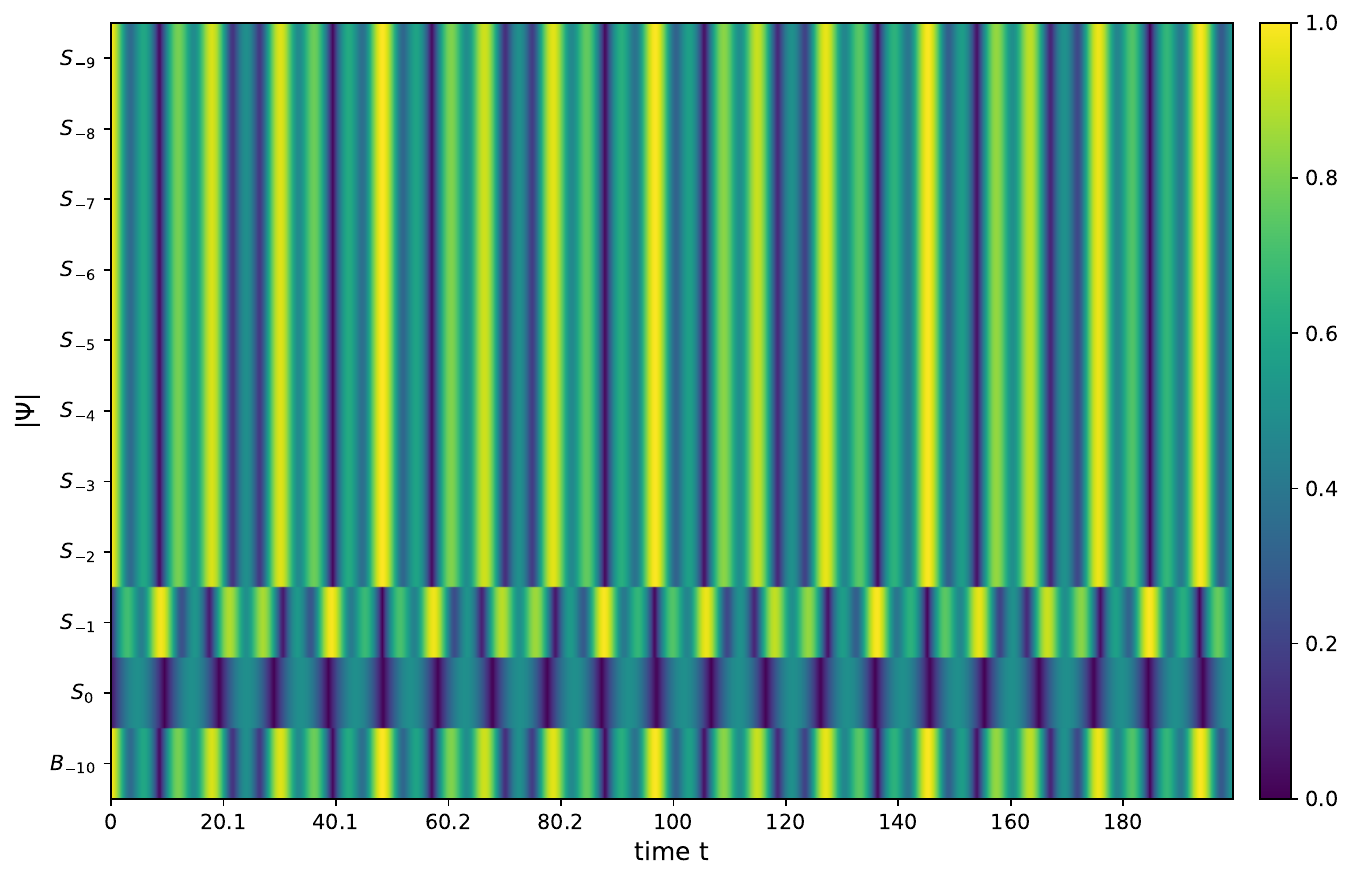}
        \caption{Amplitude for $v=56$, $W=0$, $Z=0$}
    \end{subfigure}
    \hfill
    \begin{subfigure}{0.48\textwidth}
        \includegraphics[width=\linewidth]{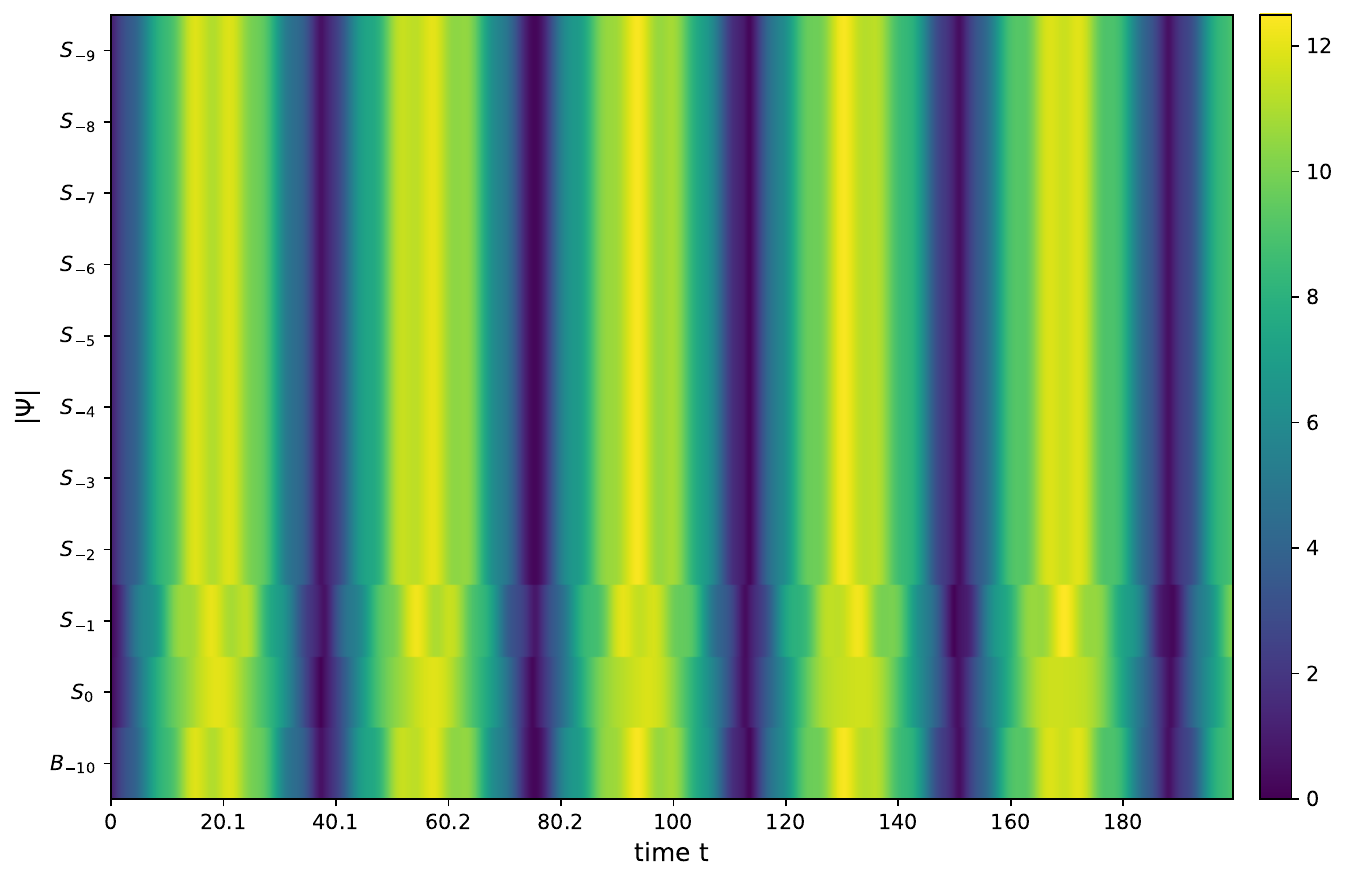}
        \caption{Amplitude for $v=56$, $W=0$, $Z=1$}
    \end{subfigure}
    
    \vspace{1em} 

    \begin{subfigure}{0.48\textwidth}
        \includegraphics[width=\linewidth]{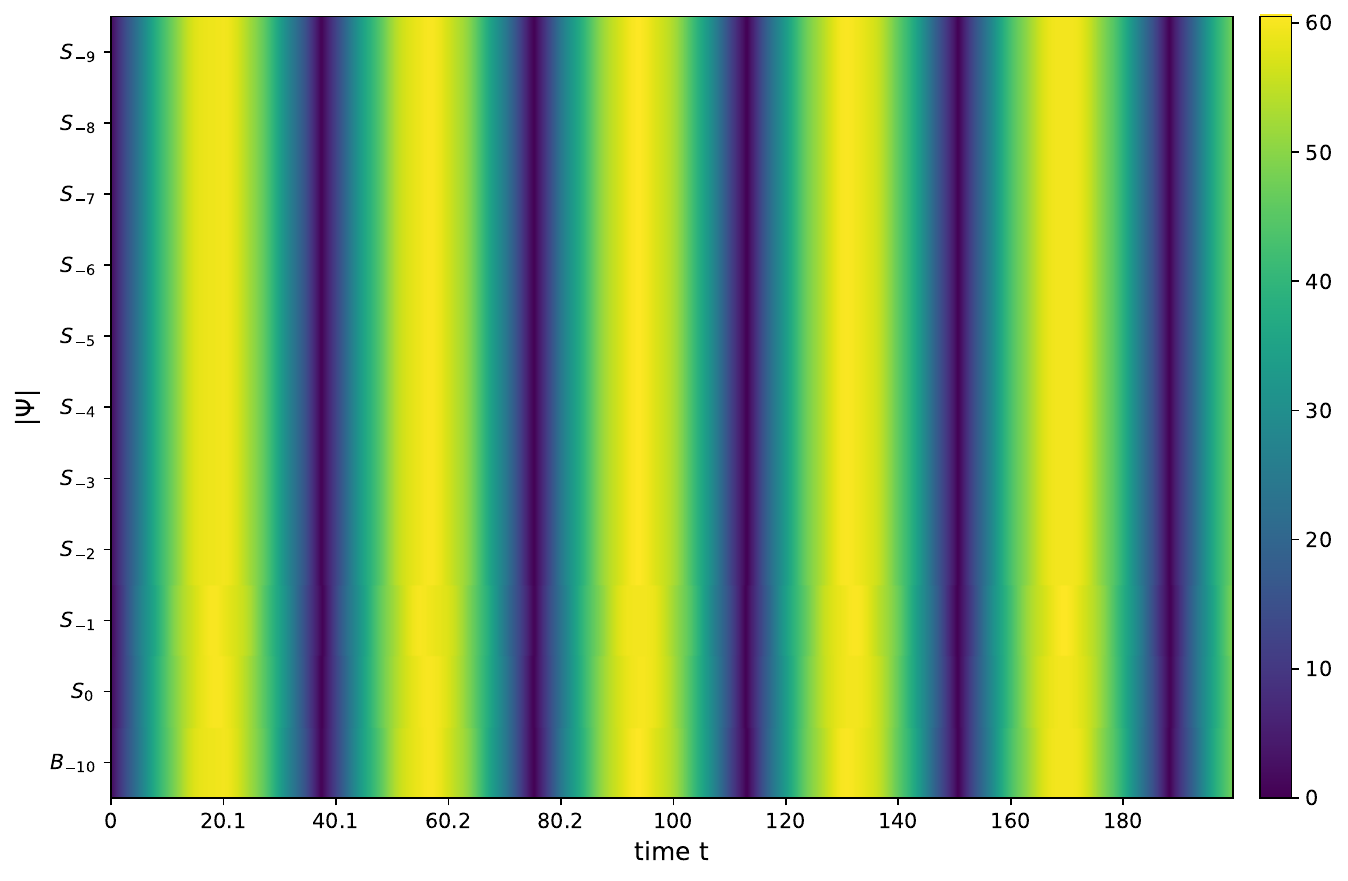}
        \caption{Amplitude for $v=56$, $W=0$, $Z=5$}
    \end{subfigure}
    \hfill
    \begin{subfigure}{0.48\textwidth}
        \includegraphics[width=\linewidth]{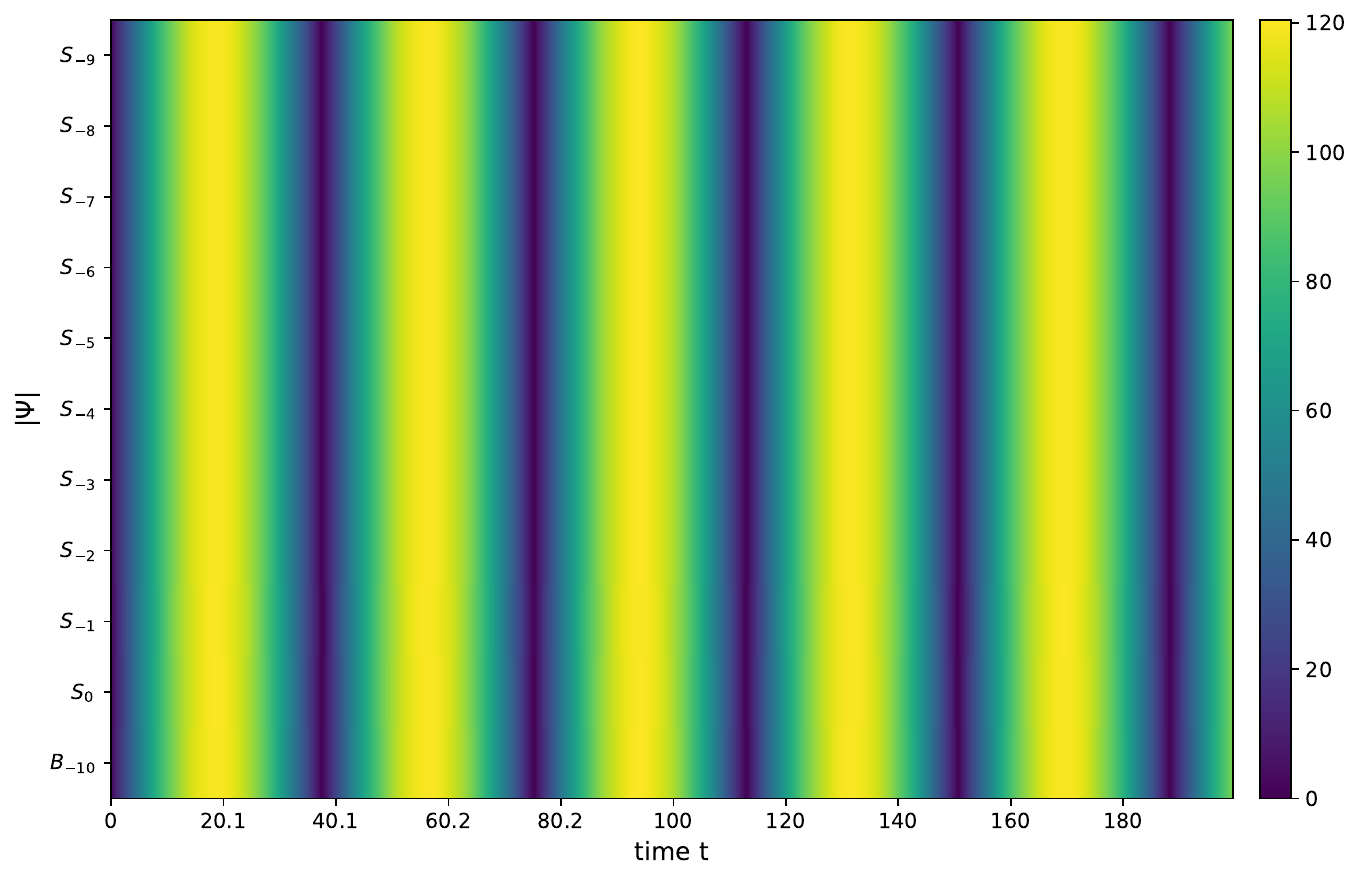}
        \caption{Amplitude for $v=56$, $W=0$, $Z=10$}
    \end{subfigure}
      \vspace{1em} 

\begin{subfigure}{0.48\textwidth}
        \includegraphics[width=\linewidth]{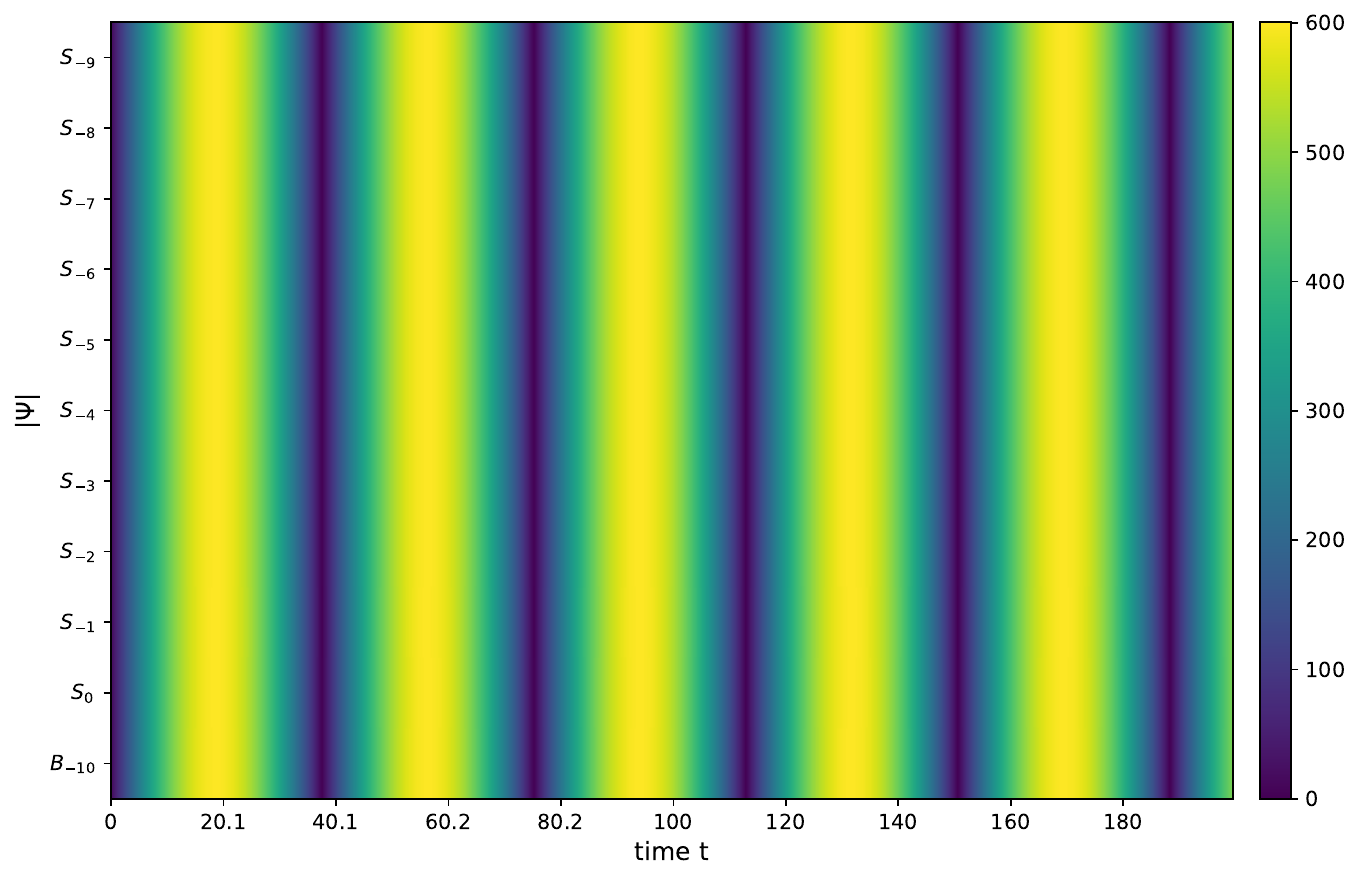}
        \caption{Amplitude for $v=56$, $W=0$, $Z=50$}
    \end{subfigure}
    \hfill
    \begin{subfigure}{0.48\textwidth}
        \includegraphics[width=\linewidth]{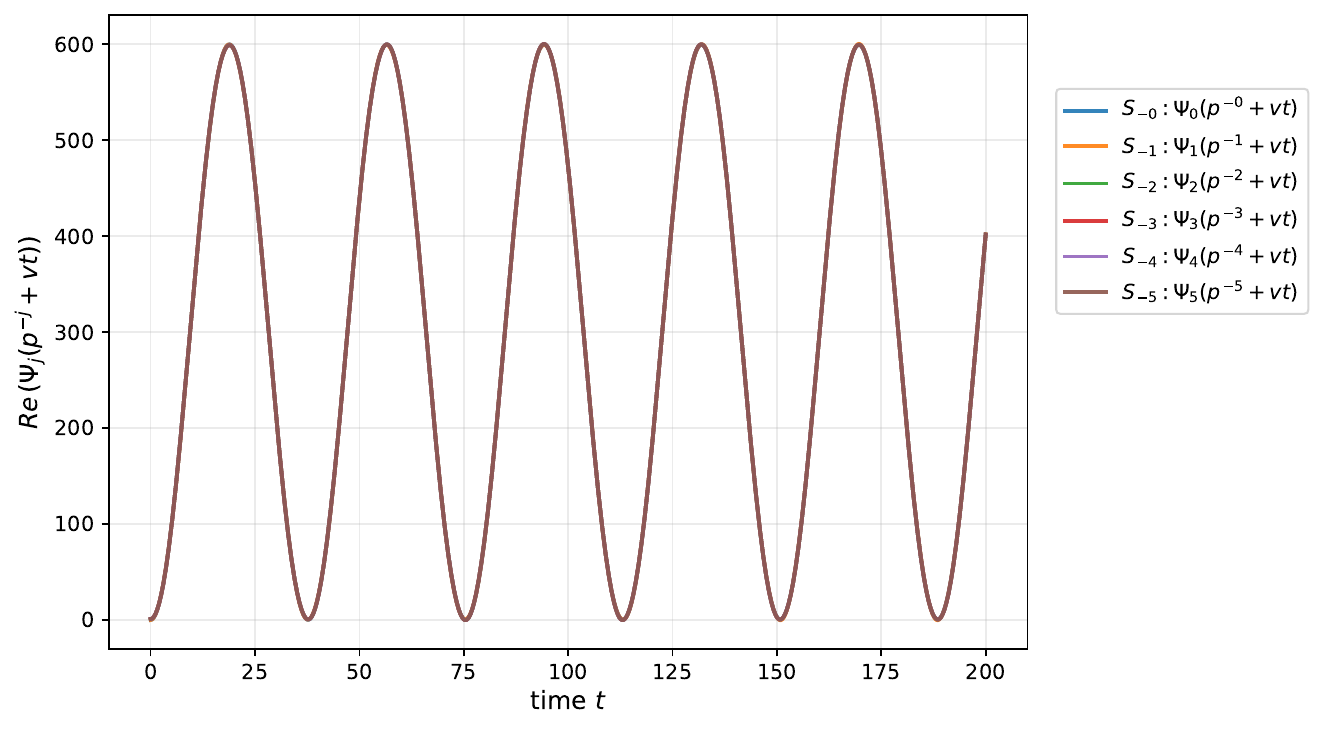}
        \caption{Profile}
    \end{subfigure}
    
    \vspace{1em} 
      
      \caption{In these simulations, $W=0$ and $v=56$ are fixed, and $Z$ is variable ($Z\in
\left\{ 0,1,5,10,50\right\} $).  For $Z\geq 5$, $\left\vert \Psi \left(
x,t\right) \right\vert $ has an oscillatory behavior independent of the
initial condition. The amplitude $\left\vert \Psi \left( x,t\right)
\right\vert $ \ grows with $Z$.  }
    \label{Figure2}
\end{figure}

The Figure \ref{Figure2} shows the simulations for $W$ fixed, $v$ fixed, and 
$Z$ variable. In these simulations, $v=56$, and $W=0$. For $Z\geq 5$, $%
\left\vert \Psi \left( x,t\right) \right\vert $ has an oscillatory behavior
independent of the initial condition. The amplitude $\left\vert \Psi \left(
x,t\right) \right\vert $ \ grows with $Z$. The Figure \ref{Figure3} shows
simulations for $W$ variable ($W\in \left\{ 0,1,10\right\} $) , $v=56$
fixed, and $Z=1$. For $W\geq 1$, $\left\vert \Psi \left( x,t\right)
\right\vert $ has an oscillatory behavior independent of the initial
condition. Furthermore, the amplitude $\left\vert \Psi \left( x,t\right)
\right\vert $ \ grows with $W$.  The Figure \ref{Figure4} shows simulations for $W$%
, $Z$ variable, and $v=56$ fixed. More precisely, 
\[
Z= 5\left\vert x\right\vert _{p}^{1.5}\text{, }W=8\left\vert
x\right\vert _{p}^{1.5},x\in \mathbb{Z}_{p}.
\]%
In this case, the values of \ $Z$ and $W$ are comparable. The amplitude $%
\left\vert \Psi \left( x,t\right) \right\vert $  has a noisy oscillatory
behavior, where the initial condition appears clearly.  In the Figure  \ref{Figure5}, each entry of the discretization of the matrix $W$ was selected as a random sample (with uniform distribution) from $[0,10]$; each entry of the discretization of the vector $Z$ was selected as a random sample (with
uniform distribution) from $[0,10]$. The velocity was \ fixed at $v=56$. The
amplitude $\left\vert \Psi \left( x,t\right) \right\vert $ exhibits noisy
oscillatory behavior depending on the initial condition.

\begin{figure}[htbp]
    \centering
    \begin{subfigure}{0.48\textwidth}
        \includegraphics[width=\linewidth]{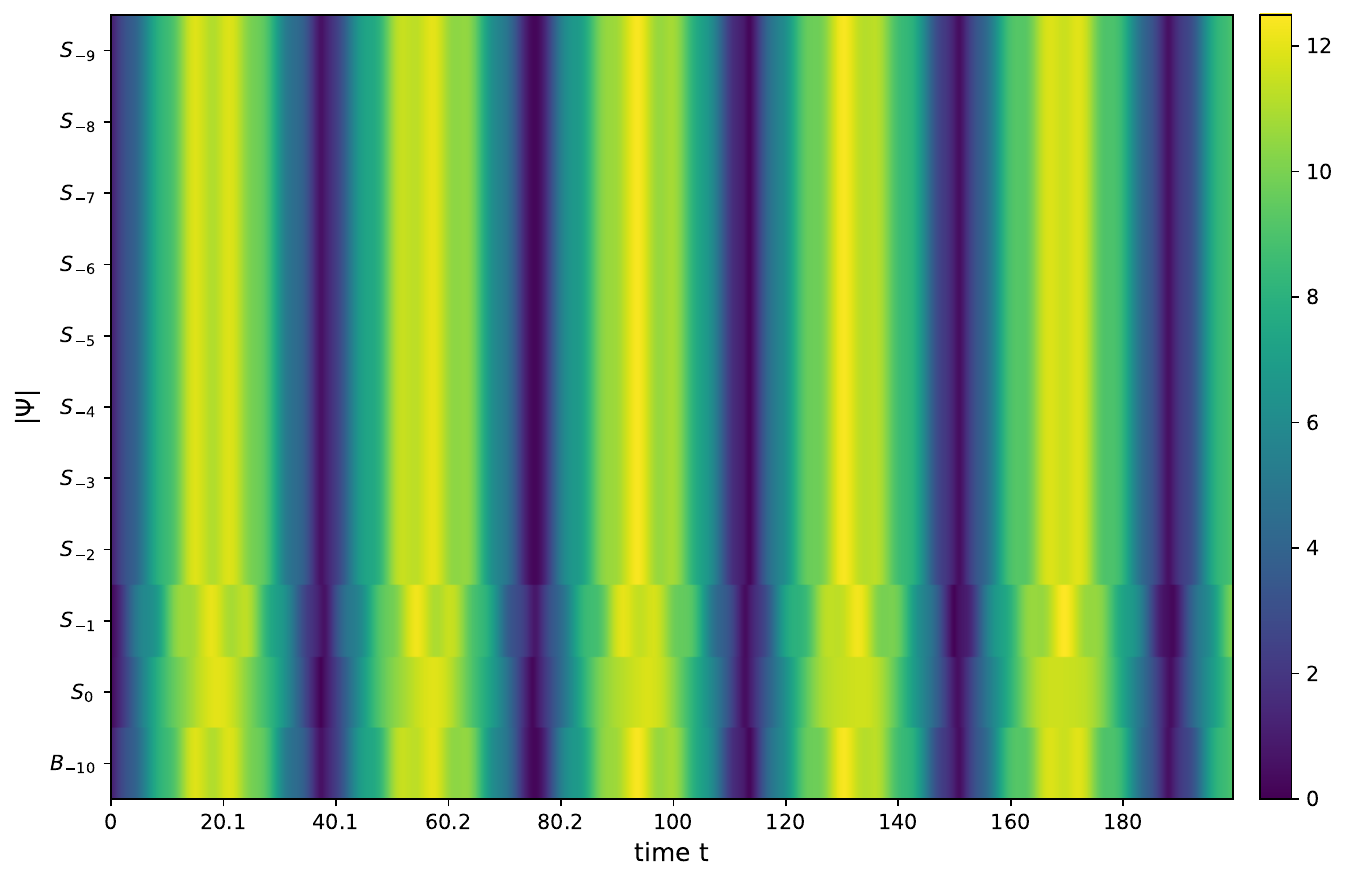}
        \caption{Amplitude for $v=56$, $Z=1$, $W=0$}
    \end{subfigure}
    \hfill
    \begin{subfigure}{0.48\textwidth}
        \includegraphics[width=\linewidth]{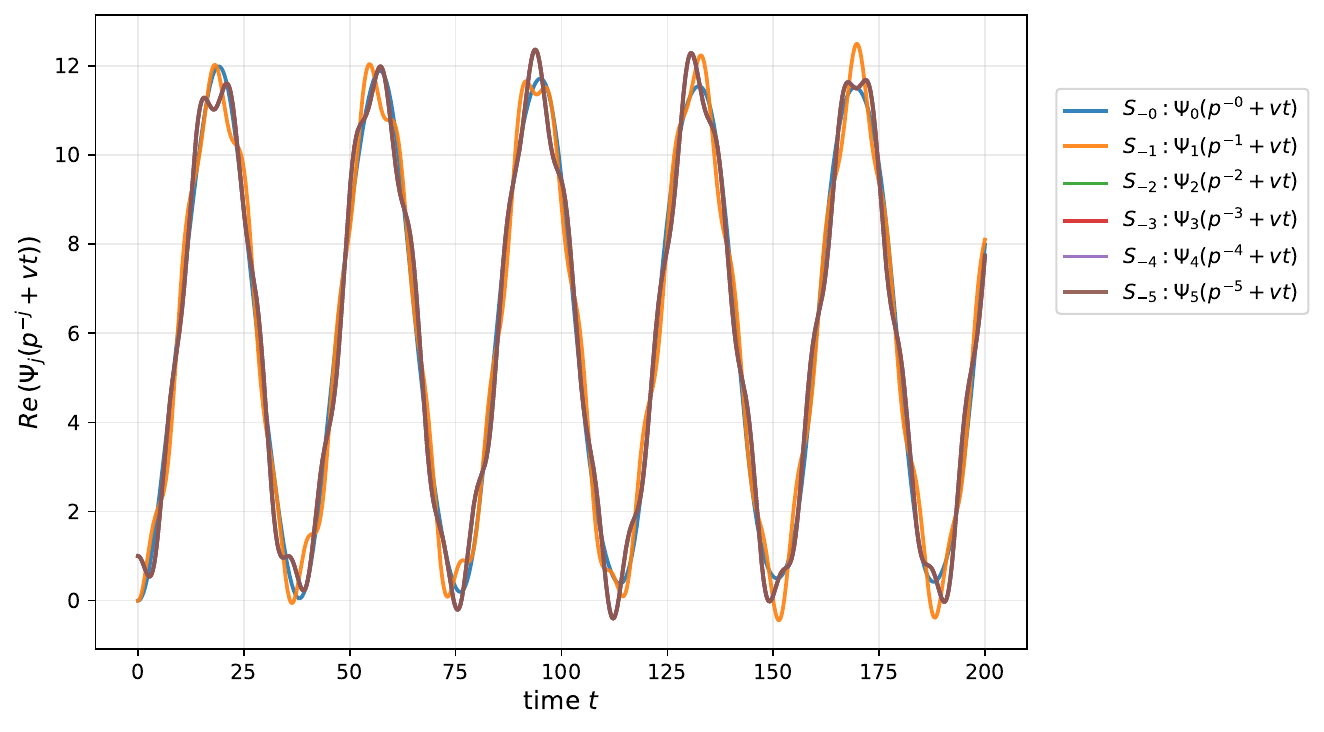}
        \caption{Profile}
    \end{subfigure}
    
    \vspace{1em} 

    \begin{subfigure}{0.48\textwidth}
        \includegraphics[width=\linewidth]{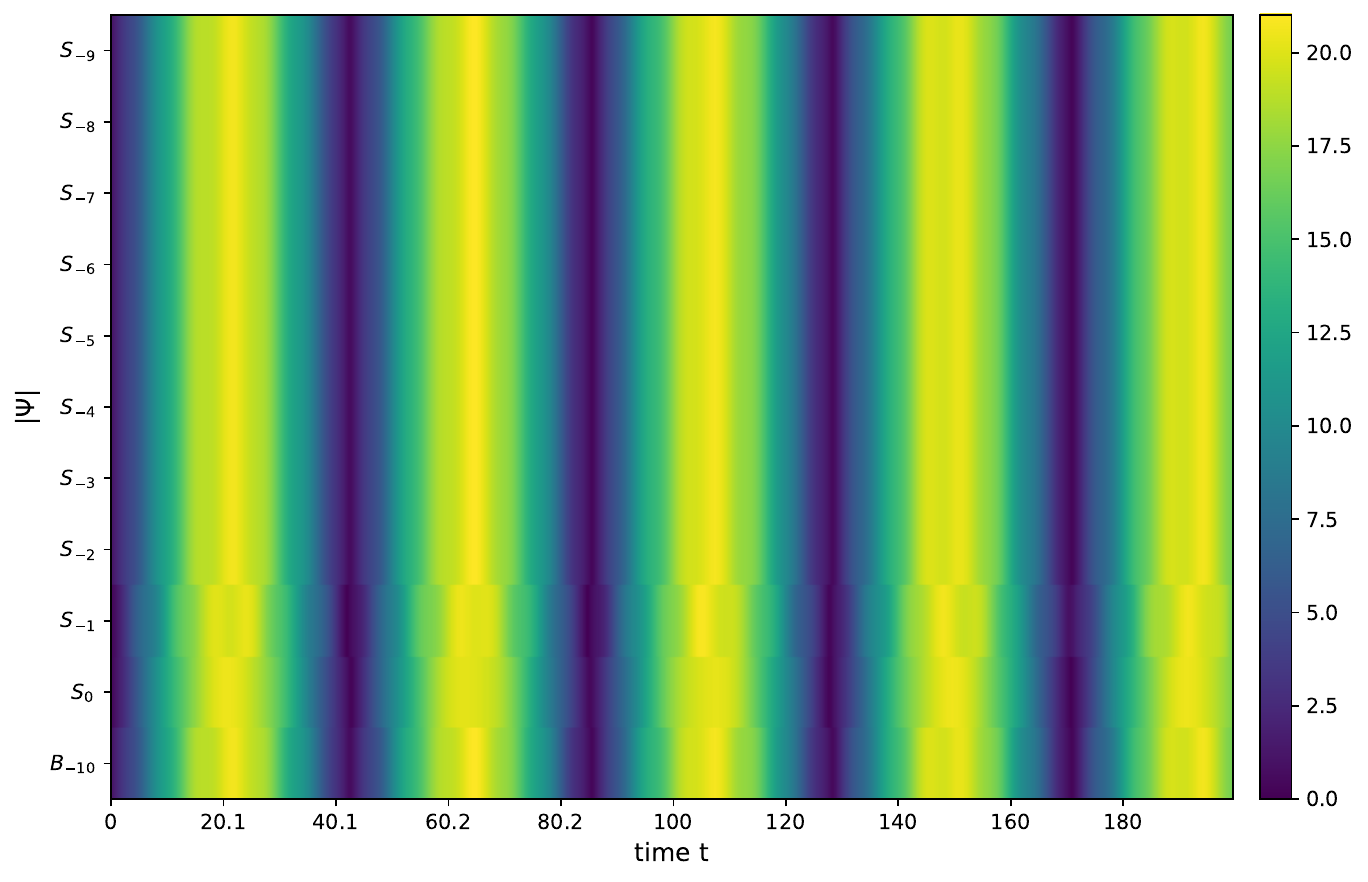}
        \caption{Amplitude for $v=56$, $Z=1$, $W=1$}
    \end{subfigure}
    \hfill
    \begin{subfigure}{0.48\textwidth}
        \includegraphics[width=\linewidth]{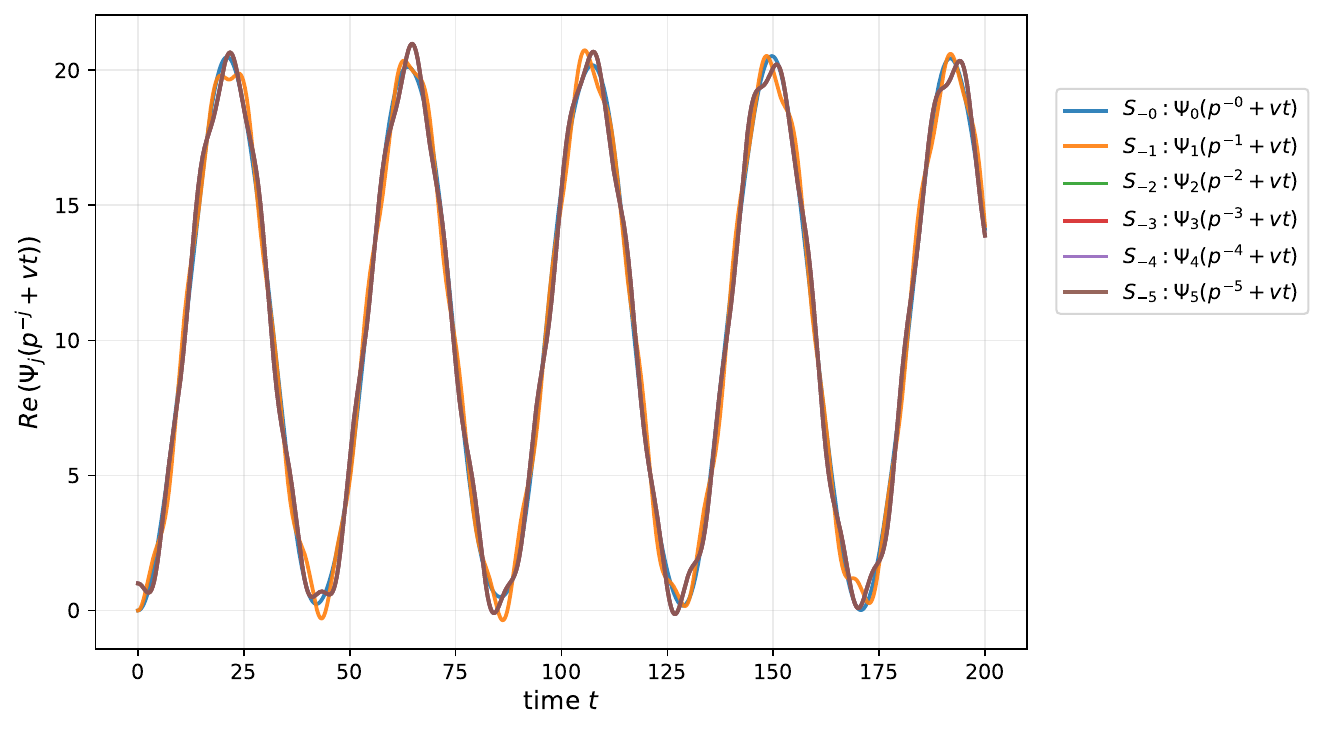}
        \caption{Profile}
    \end{subfigure}
    
    \vspace{1em} 

    \begin{subfigure}{0.48\textwidth}
        \includegraphics[width=\linewidth]{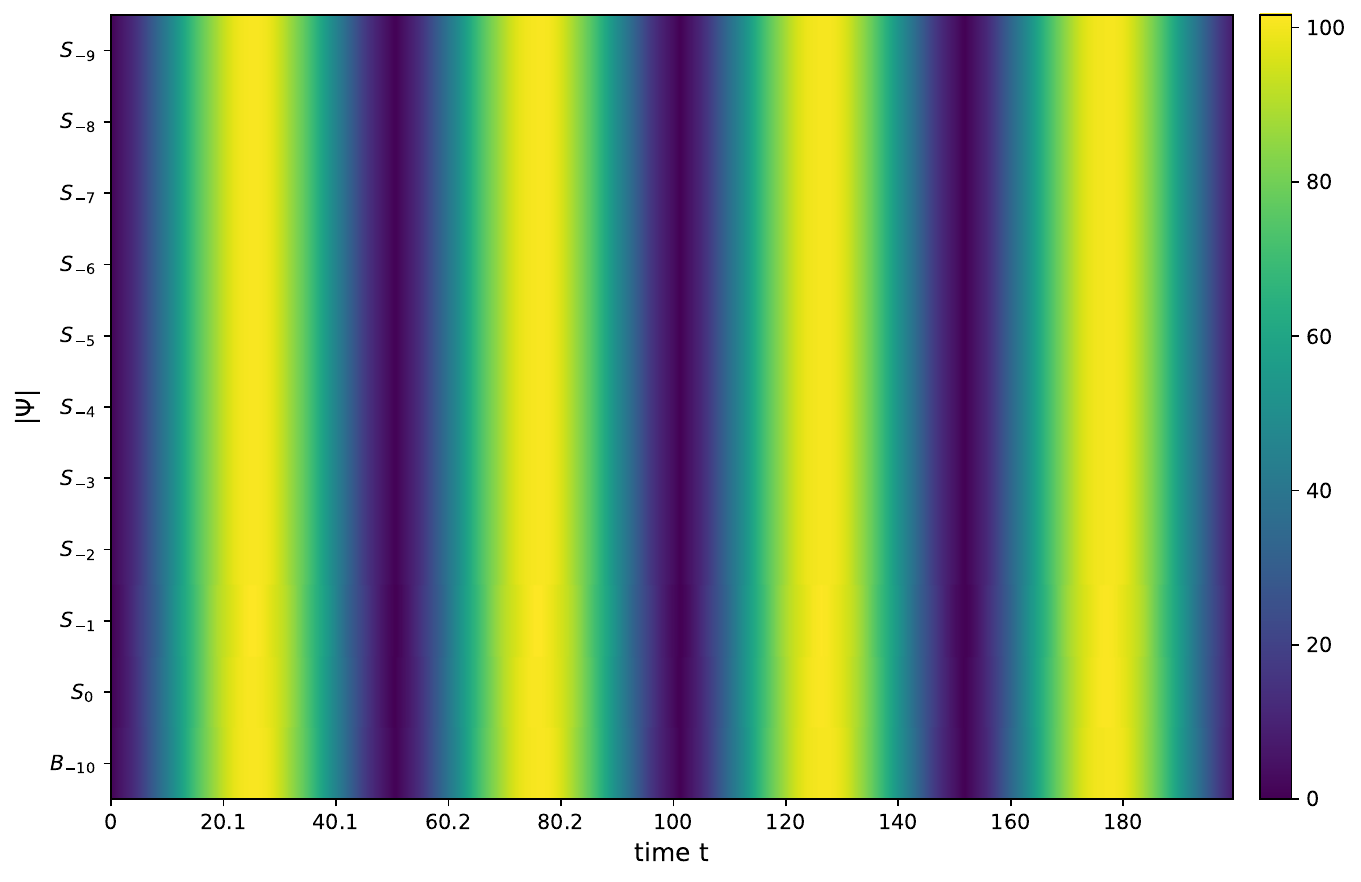}
        \caption{Amplitude for $v=56$, $Z=1$, $W=10$}
    \end{subfigure}
    \hfill
    \begin{subfigure}{0.48\textwidth}
        \includegraphics[width=\linewidth]{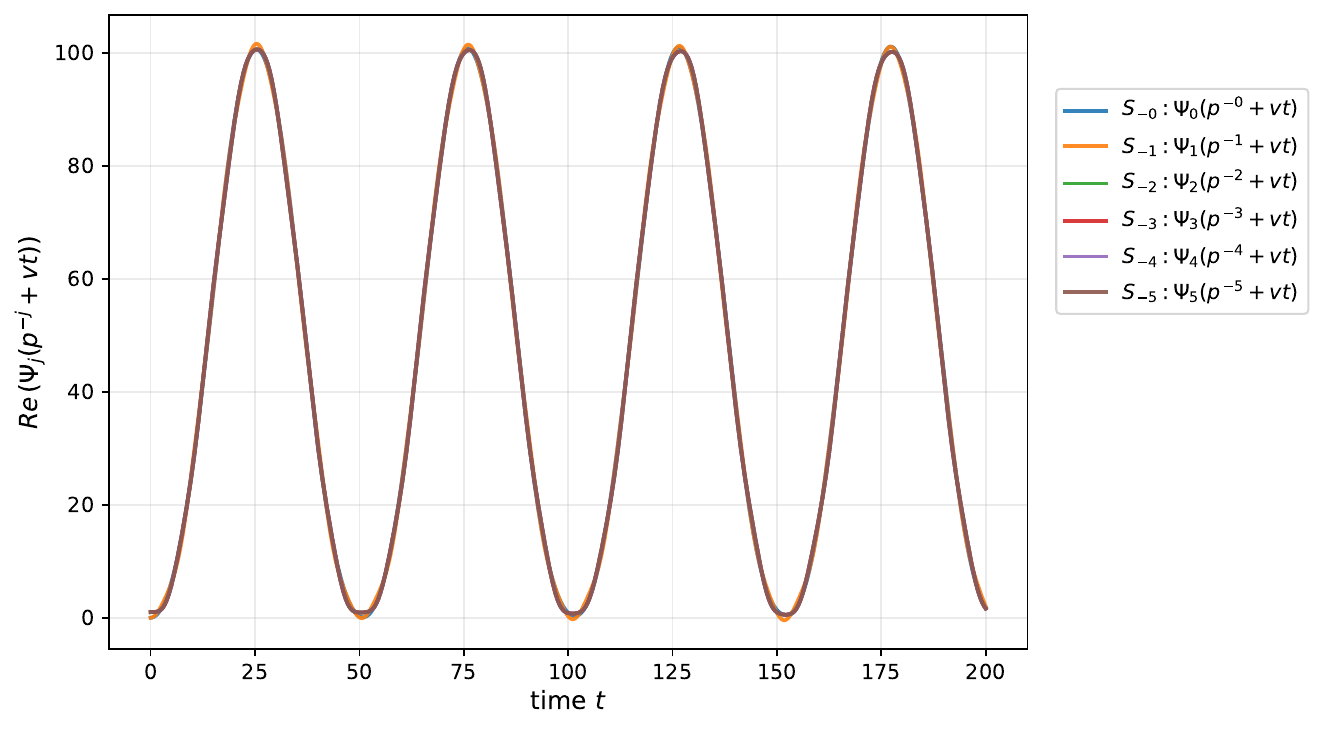}
        \caption{Profile}
    \end{subfigure}
    
    \vspace{1em} 

    \caption{In these simulations,  $W$ is variable ($W\in \left\{ 0,1,10\right\} $) ,
while $v=56$ and $Z=1$ are fixed. For $W\geq 1$, $\left\vert \Psi \left(
x,t\right) \right\vert $ has an oscillatory behavior independent of the
initial condition. Furthermore, the amplitude $\left\vert \Psi \left(
x,t\right) \right\vert $ \ grows with $W$. \ The profiles  show that the $%
\mathrm{Re}\left( \Psi \left( x,t\right) \right) $ are oscillations whose
amplitudes grow with $W$.
}
    \label{Figure3}
    \end{figure}

\begin{figure}[htbp]
    \centering
    \begin{subfigure}{0.48\textwidth}
        \includegraphics[width=\linewidth]{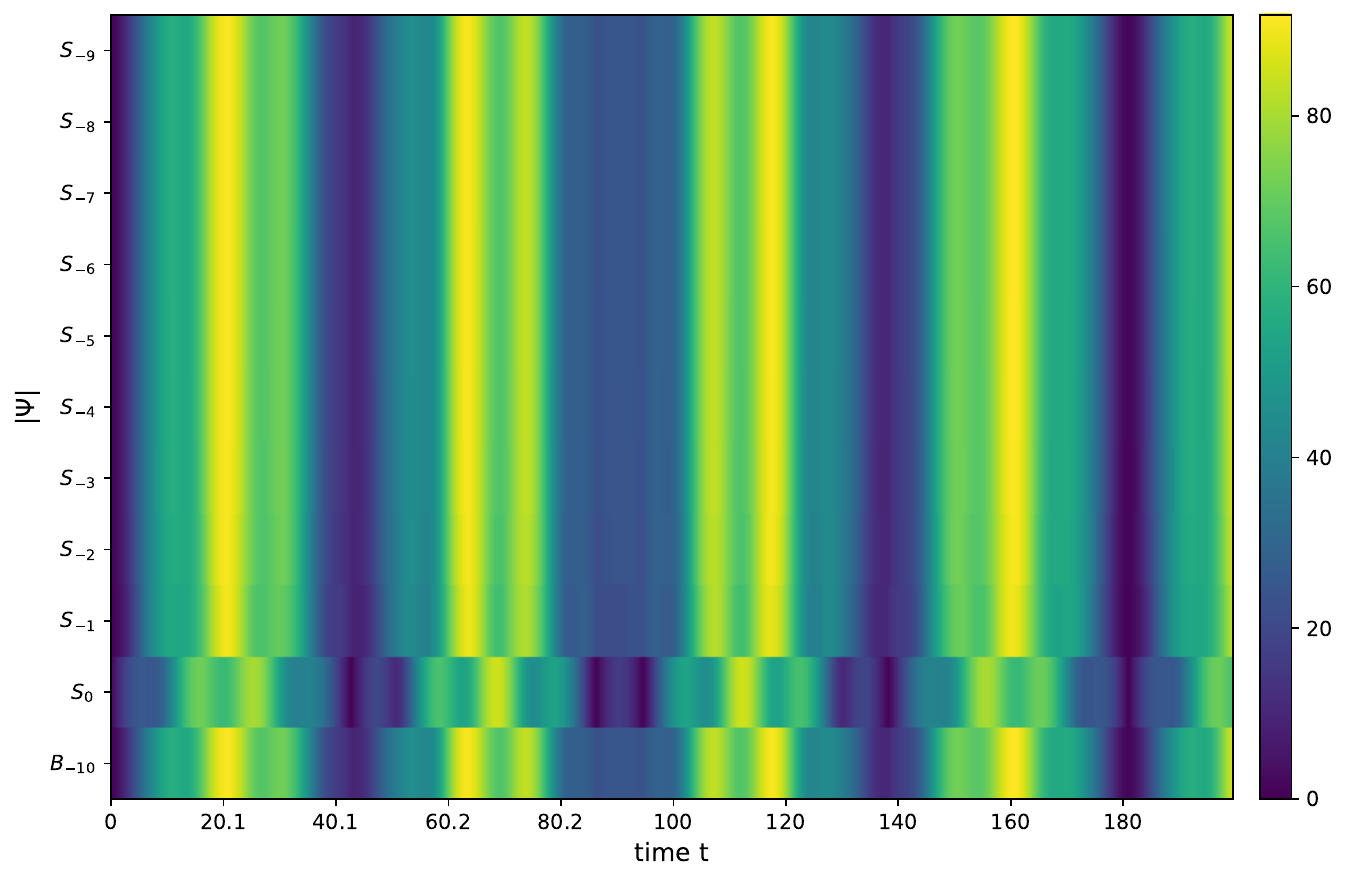}
        \caption{Amplitude }
    \end{subfigure}
  \hfill
\begin{subfigure}{0.51\textwidth}
        \includegraphics[width=\linewidth]{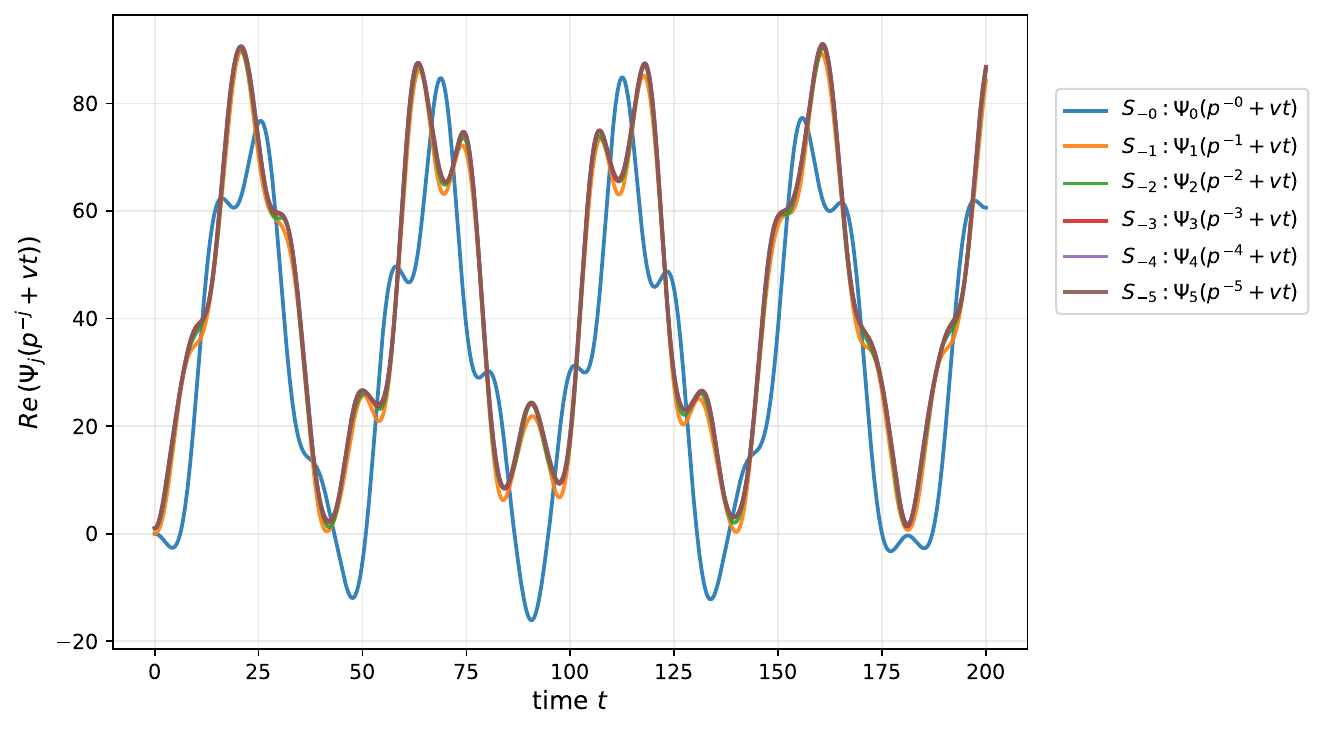}
        \vspace{-0.12cm}
        \caption{Profile}
    \end{subfigure}
    \caption{The parameters are fixed as $v=56$, $Z=\bigskip 5\left\vert x\right\vert
_{p}^{1.5}$, $W=8\left\vert x\right\vert _{p}^{1.5}$, $x\in \mathbb{Z}_{p}$.
In this case the values of $Z$ and $W$ are comparable, and the amplitude $%
\left\vert \Psi \left( x,t\right) \right\vert $ exhibits a noisy oscillatory
behavior, where the initial datum can be seen clearly. The noisy oscillatory
behavior is clearly illustrated in the amplitude profile.
}
\label{Figure4}
\end{figure}

\begin{figure}[htbp]
    \centering
    \begin{subfigure}{0.48\textwidth}
        \includegraphics[width=\linewidth]{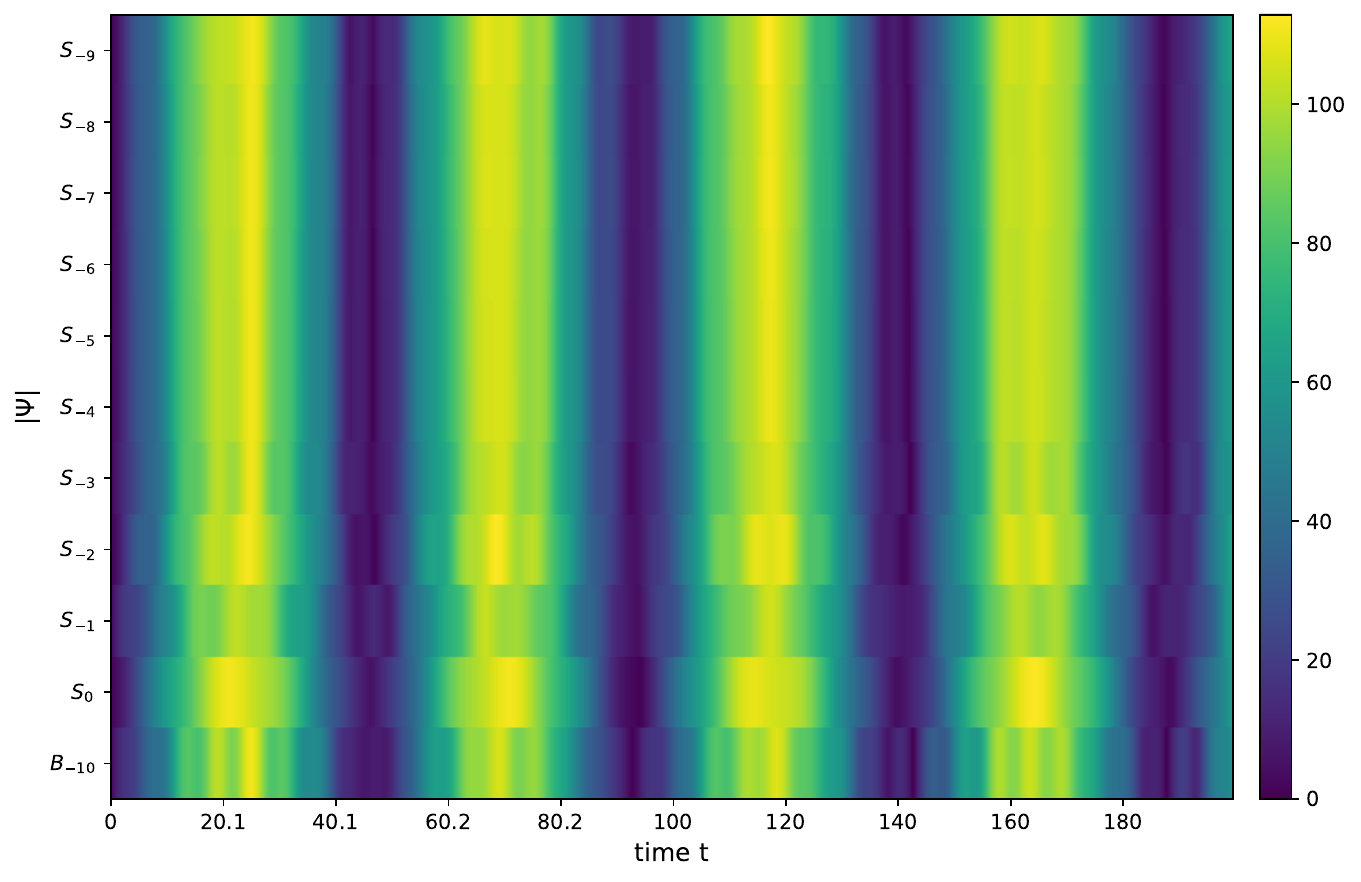}
        \caption{Amplitude $\left\vert \Psi \left( x,t\right) \right\vert $
}
    \end{subfigure}
  \hfill
\begin{subfigure}{0.51\textwidth}
        \includegraphics[width=\linewidth]{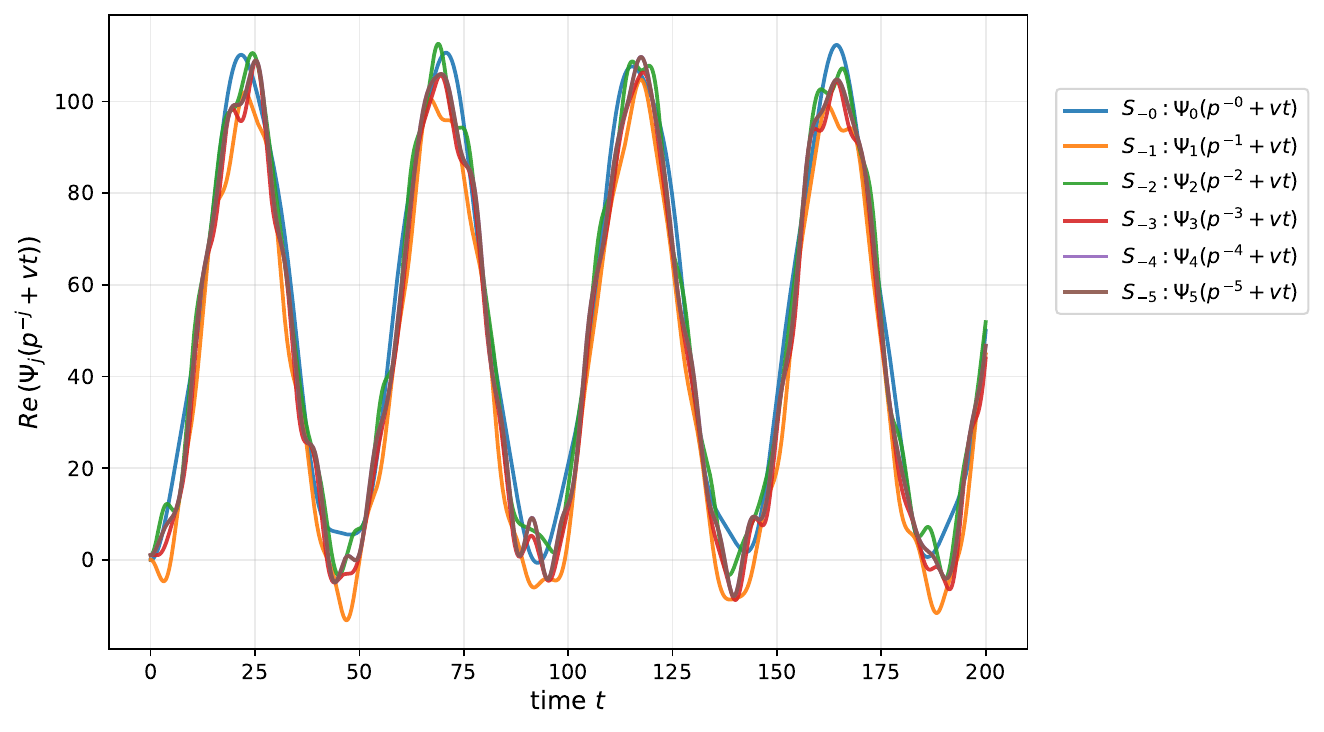}
        \vspace{-0.12cm}
        \caption{Profile}
    \end{subfigure}
    \caption{The velocity is fixed at $56$. The discretization of $W$ is a matrix $\left[
W_{ij}\right] $, $i,j\in \mathbb{Z}_{2}/2^{10}\mathbb{Z}_{2}$, where $W_{ij}$
is a sample of a uniformly distributed random variable in the interval $%
\left[ 0,10\right] $. The discretization of $Z$ is a vector $\left[ Z_{i}%
\right] $,  $i\in \mathbb{Z}_{2}/2^{10}\mathbb{Z}_{2}$, where $Z_{i}$ is a
sample of a uniformly distributed random variable in the interval $\left[
0,10\right] $. The amplitude $\left\vert \Psi \left( x,t\right) \right\vert $
exhibits a noisy oscillatory behavior, where the initial datum can be seen
clearly. The profile of  $\Psi \left( x,t\right) $, corresponding to $\mathrm{Re}(\Psi \left( x,t\right) )$ evaluated on some spheres, clearly shows noisy
oscillations.}
\label{Figure5}
\end{figure}

\end{document}